\PassOptionsToPackage{svgnames}{xcolor}
\documentclass[12pt]{article}

\usepackage{amsmath,amssymb,amsfonts,amsthm,setspace,tikz,natbib,dsfont, xcolor, array,colortbl,mathrsfs,pifont,wasysym,dirtytalk}
\usetikzlibrary{calc}
\usepackage{pgfplots}
\pgfplotsset{compat=1.7}
\usepgfplotslibrary{ternary, units}
\usepackage{tikz}
\usepackage{tikz-3dplot}
\usetikzlibrary{decorations.pathmorphing, pgfplots.ternary, pgfplots.units}
\usetikzlibrary{positioning}
\usepackage{centernot}
\usepackage{xspace}
\usepackage{mathtools}
\usepackage[scr=boondox,scrscaled=1.05]{mathalfa}

\usepackage{graphicx}
\usepackage{subcaption}

\usepackage{thmtools,thm-restate}
\usepackage{enumitem}

\usepackage[makeroom]{cancel}
\usepackage{bm}

\usepackage[letterpaper]{geometry}
\definecolor{lam1}{HTML}{001bb0}
\definecolor{lam2}{HTML}{bc0112}
\definecolor{lam3}{HTML}{3d2978}

\usepackage[colorlinks=true,urlcolor=lam2,linkcolor=lam2,citecolor=lam3]{hyperref}

\onehalfspace

\newtheoremstyle{mytheoremstyle} 
{\topsep}                    
{\topsep}                    
{\itshape}                   
{}                           
{\sc}                   
{.}                          
{.5em}                       
{}  

\theoremstyle{mytheoremstyle}

\newtheorem{theorem}{Theorem}
\newtheorem{lemma}{Lemma}

\newtheoremstyle{scfont} 
{\topsep}                    
{\topsep}                    
{}                   
{}                           
{\scshape}                   
{}                          
{.5em}                       
{\textbf{Axiom \thmnumber{#2}\thmname{#1}}\thmnote{---{\textsc{#3}.}}}

\theoremstyle{scfont}

\theoremstyle{remark}

\newenvironment{tproof}[1]{\noindent \emph{Proof of Theorem \ref{#1}. \phantomsection\label{pf:#1}\addcontentsline{toc}{subsection}{Pf. Thm. \ref{#1}}}}{\hfill$\blacksquare$ \\}

\newenvironment{subproof}[1][\proofname]{
	\begin{proof}[#1]%
	}{%
	\end{proof}%
}

\newcounter{ax}
\newenvironment{ax}[3]{
	\medskip
	\addtocounter{ax}{1}
	\noindent \textbf{#1\theax}---\sclabel{#2}{#3}.
}{
	\medskip
}

\newcounter{ex}

\usepackage{titlesec}

\titleformat*{\section}{\large\scshape\centering}
\titleformat*{\subsection}{\scshape\centering}
\titleformat*{\subsubsection}{\itshape}
\titleformat*{\paragraph}{\bfseries\centering}
\titleformat*{\subparagraph}{\large\bfseries\centering}

\titlespacing*{\section}{0pt}{5.5ex plus 1ex minus .2ex}{3.5ex plus .2ex}

\renewcommand{\emptyset}{\varnothing}
\renewcommand{\phi}{\varphi}

\def \D{\mathbb{D}}
\def \E{\mathbb{E}}
\def \R{\mathbb{R}}
\def \Pr{\mathbb{P}}

\def \NN{\mathcal{N}}
\def \WW{\mathcal{W}}
\def \VV{\mathcal{V}}
\def \W{\Omega}

\def \s{\succcurlyeq}
\def \P{\mathscr{P}}
\def \Q{\mathscr{Q}}
\def \H{\mathscr{H}}
\def \C{\mathscr{C}}

\def \face{\mathbb{F}}
\def \K{\mathcal{K}}
\def \deg{\emptyset}
\def \U{\mathcal{U}}
\def \vnm{\textup{\texttt{vnm}}}

\def \<{\langle}
\def \>{\rangle}
\newcommand{\mathbbm}[1]{\text{\usefont{U}{bbm}{m}{n}#1}}
\def \1{\mathbbm{1}}
\def \0{\bm{0}}

\DeclareMathSymbol{\smm}{\mathbin}{AMSa}{"39}

\let\integral\int
\renewcommand{\int}{\textup{int}}
\def \ri{\textup{ri}}
\def \conv{\textup{conv}}
\def \cl{\textup{cl}}
\def \ext{\textup{ext}}
\def \UAP{W_{A,P}}
\def \supp{\textup{supp}}
\def \part{\textup{\texttt{part}}}
\def \poly{\textup{\texttt{poly}}}
\def \prob{\textup{\texttt{prob}}}

\makeatletter
\newcommand*\bigcdot{\mathpalette\bigcdot@{.6}}
\newcommand*\bigcdot@[2]{\mathbin{\vcenter{\hbox{\scalebox{#2}{$\m@th#1\bullet$}}}}}
\makeatother

\newcommand{\hd}[1]{\medskip \noindent \textbf{#1} \  $\bigcdot$  }

\makeatletter
\newcommand{\mylabel}[3]{\def\@currentlabel{#2}\phantomsection\textsc{#3} (\texttt{#2})\label{#1}}
\newcommand{\sclabel}[3]{\def\@currentlabel{\theax}\phantomsection\lowercase{\textsc{#1}.}\label{ax:#2}}
\makeatother

\newcommand{\rexp}[1]{\hyperref[exp:#1]{\textcolor{lam1}{\textsc{exp~#1}}}}
\renewcommand{\r}[1]{\hyperref[#1]{\textcolor{lam1}{\textup{\textbf{A\ref{#1}}}}}}
\newcommand{\rax}[1]{\hyperref[ax:#1]{\textup{(\textsc{#1})}}}

\DeclareMathOperator*{\argmax}{arg\,max}

\title{Modeling the Modeler: A Normative Theory of Experimental Design\footnote{We thank seminar participants at Bart Lipman's Conference, Autonomous University of Barcelona and Centre d'\'Economie de la Sorbonne.}}
\author{
	Fernando Payr\'{o}\footnote{Autonomous University of Barcelona, Barcelona School of Economics and Center for the Study of Organizations and Decisions in Economics; \href{mailto:fernando.payro@uab.cat}{\tt fernando.payro@uab.cat}. Payr\'o gratefully acknowledge financial support from the Ministerio de Economía y Competitividad and Feder (PID2020-116771GB-I00 and PID2023-147183NB-I00) and financial support from the Spanish Agencia Estatal de Investigación (AEI), through the Severo Ochoa Programme for Centres of Excellence in R\&D (CEX2019-000915-S)}  \ and
	Evan Piermont\footnote{Royal Holloway, University of London, Department of Economics; \href{mailto:evan.piermont@rhul.ac.uk}{\tt evan.piermont@rhul.ac.uk}.}
} 
\begin{document}
	
	\maketitle
	
	
	\begin{abstract}
		We consider an analyst whose goal is to identify a subject's utility function through revealed preference analysis. We argue the analyst's preference about which experiments to run should adhere to three normative principles: The first, \emph{Structural Invariance}, requires that the value of a choice experiment only depends on what the experiment may potentially reveal.
		The second, \emph{Identification Separability}, demands that the value of identification is independent 
		of what would have been counterfactually identified had the subject had a different utility. Finally, \emph{Information Monotonicity} asks that more informative experiments are preferred.
		We provide a representation theorem, showing that these three principles characterize \emph{Expected Identification Value} maximization, a functional form that unifies several theories of experimental design. We also study several special cases and discuss potential applications.       
		\footnotesize

		\vspace{0.2cm}
		
		\textsc{Declarations of interest}: None
	\end{abstract}
	
	
	\newpage
	
	\section{Introduction}
	This paper proposes a theory of experimental design. We consider an analyst who's goal is to identify a subject's utility function through the use of experiments, specifically, by offering the subject decision problems and observing his choices. Because of time, cost, or computational constraints, the analyst will only be able to offer a limited set of decision problems; how then should she choose which experiments to conduct? In this paper, we advance three normative principles and argue that they should guide any rational experimental design, independent of the specific objectives of the analyst.
	
	As we explain in detail below, each observation in an experiment (partially) identifies some set of utilities, those that that are consistent with the observation. For example, observing a subject known to be a CRRA utility maximizer reject an actuarially fair lottery would identify him as risk averse, but might not resolve anything further about their coefficient of risk aversion. With this notion of partial identification in mind, the three normative principles can be stated as follows: First, \emph{Information Monotonicity}, asserts that the analyst prefers sharper identification; that is, if one experiment always identifies a smaller subset of utilities than another experiment, it is preferred. 
	Second, \emph{Structural Invariance}, maintains that the value of an experiment should depend only on what it allows the analyst to identify and not on other structural details. Specifically, if two experiments yield the same set of possible inferences about the subject, they must be valued equally. 
	Finally, \emph{Identification Separability} demands that the value of making some partial identification should depend only on what was identified, and not on what the experiment would have counterfactually identified had the subject made a different choice. 
	
	Our main result shows that an analyst's preference over experiments adheres to these principles if and only if she seeks to maximize \emph{expected identification value}, as we now explain. A rational analyst should be able to assign a value to each partial identification, that is, a value for learning that the subject's true utility lies is some subset. For example, an analyst can reflect on the relative value of learning a CRRA subject  ``is risk averse'' versus learning the subject ``has a coefficient of risk aversion in $[1.5,2]$.'' This value of identification is subjective and encodes the analyst's particular goals and motivations.%
	\footnote{For instance, an analyst may only be only interested in classifying subjects as risk averse or not, but uninterested in any further details of the subject's preference. Such an analyst would be indifferent between the two partial identifications from the prior sentence.} 
	
	Given a ex-ante distribution over the utility type of the subject, interpreted as the analyst's prior belief about the subject's preferences, each experiment will induce a distribution over consequent partial identifications: our three normative principles characterize ranking experiments in accordance with the expected value of the identification they permit.
	By accommodating any subjective value of identification, our theory unveils the common facets of rational experimental design that are independent of the idiosyncratic objectives of the analyst. Indeed, we detail how our approach unifies several distinct experimental paradigms. We then show how this theory provides additional insight in specific environments, where the normative principles settle  concrete design choices. In particular, we specialize the model in two ways. First, by examining the case where the subject is known to be an expected utility maximizer. Second, by examining the case where the analyst seeks to maximize the reduction in entropy between her prior and posterior.
	
	\hd{Discussion of Model and Results}
	To keep the analysis simple, we abstract away from the physical details of an experiment. Instead, we model an experiment as a menu of alternatives $A$ and a partition of the menu $\P$ that captures what is observable to the analyst; when the subject chooses $a\in A$, the analyst observes the (unique) $P\in \P$ such that $a \in P$.  When $\P$ is the discrete partition, the subject's choice is perfectly observed, as is likely in very simple environments, e.g., single-stage laboratory experiments. By allowing $\P$ to be coarser, we allow for more complex experimental environments. For example, $A$ could represent the set of strategies in a dynamic environment where $\P$ captures the unobservability of off-path behavior.\footnote{This interpretation is discussed in Section \ref{sec:partial-observability}.} In less controlled environments, observational restrictions are common-place, e.g., an online platform (google) might be able to observe which retailer was chosen by a user, but not the user's actual purchase.  We take as primitive a preference over \emph{randomized experiments}, that is, finitely supported distributions $\pi$ over experiments. In the literature, such experiments are called discrete choice experiments. 
	We also take as part of our primitive
	the set of ex-ante possible utility functions from which the subject's true utility, $u^*$, was drawn and  
	a probability $\mu$ over $\U$, capturing the prior beliefs of the analyst.\footnote{As this is a normative exercise, we are interested in providing a potential experimenter with guidance on how to construct rational preferences, rather than identifying her prior beliefs; as such we take $\mu$ as part of the primitive. Our theory can be applied, modulo certain technicalities, in the event there is no prior beliefs, as discussed in Section \ref{sec:belief-free}.} 
	
	
	Given an experiment $(A,\P)$ and an observation $P \in \P$, let $\UAP \subseteq \U$ denote those utility functions that would choose an element in $P$ out of $A$. Thus, the observation that an element of $P$ was chosen out of $A$ induces the partial identification of $\UAP$, i.e., the analyst infers that $u^* \in \UAP$.   Our main result shows that Structural Invariance, Identification Separability, and Information Monotonicity\footnote{Along with an an axiom dictating that the analyst is an expected utility maximizer with respect to the randomization across experiments.} hold if and only if the analyst assigns a value, $\tau(W)$, to each possible partial identification, $W \subseteq \U$, and  values experiments according to the expected value of the identification they will yield. Formally, the analyst's preference must be representable by a \emph{expected identification value} functional of the form
	\begin{equation}
		\label{eq:rep2}
		\tag{$\star$}
		F(\pi)=\sum_{\supp(\pi)}\Big(\sum_{P \in \P} \tau(W_{A,P}) \mu(W_{A,P})\Big)\pi(A,\P)
	\end{equation}
	where $\pi$ is a lottery over experiments,  $\tau$ is interpreted as an \emph{identification index} and $\mu$ as the analysts prior.
	
	Since $\tau$ depends only on what can be inferred from the observed outcome, each experiment is equated with the sets of utilities it can partially identify. This reflects Structural Invariance. Moreover, the representation is additive across cells of the partitions and thus, 
	the ranking between two experiments is independent of whatever they commonly identify.
	This reflects Identification Separability. Finally, Information Monotonicity, requires the identification index $\tau$ to always find information weakly beneficial: for all disjoint $W,W' \subseteq \U$ (set $V = W \cup W'$) it must be that 
	$$\tau(W)\mu(W | V) + \tau(W')\mu(W' | V) \geq \tau(V)$$
	where $\mu(\cdot|V)$ is the conditional of $\mu$ given $V$.
	That is, given that the analyst can already identify $V$, the expected value of further learning the distinction between $W$ and $W'$ is weakly positive. 
	
	
	
	\hd{Normative Principles as Guiding Design Choices} To better understand how Structural Invariance and Identification Separability relate to experimental design choices, we now consider two simple examples within the environment of expected utility maximizing subjects. We apply our general theory to this setting in Section \ref{sec:EUmodel}.
	
	In particular, assume that the subject entertains a linear utility function over lotteries defined on three alternatives $\{a,b,c\}$. Denote by $\alpha x+(1-\alpha)y$ the lottery that places probability $\alpha$ on alternative $x$ and $(1-\alpha)$ on $y$, for $x,y \in \{a,b,c\}$. 
	Consider an analyst who needs to choose one of the following two experimental procedures. Both experimental procedures offer two menus to the subject:
	\begin{align*}
		\textsc{exp a}: A\phantom{'}&= \{a, \ \ \tfrac12 a + \tfrac12 b, \ \   \tfrac12 a + \tfrac12c, \ \ \tfrac12 b + \tfrac12c\} 
		\\
		A'&= \{\tfrac{6}{10}b + \tfrac{4}{10}c, \ \  \tfrac{4}{10}b + \tfrac{6}{10}c\}.
		\\
		\textsc{exp b}: B\phantom{'}&= \{a, \ \ b, \ \ c\}
		\\
		B'&=  \{\tfrac23 a + \tfrac13b, \ \ \tfrac23 a + \tfrac13c, \ \ \tfrac13 a + \tfrac13b + \tfrac13c\}.
	\end{align*}
	These are visualized in the top of Figure \ref{fig:normalconesex}. 
	
	Which of these two experiments should the analyst run? At first glance, this appears to be a matter of taste, as it seems plausible the answer should depend on the analyst's objectives, that is, on which aspects of the subject's preference she is interested in identifying. However, the principle of Structural Invariance imposes that these two experiments must be valued equally, as they induce the same possible set of partial identifications. To see this, observe that because expected utility is linear, observing a choice from $A$ and $A'$ is informationally equivalent to observing a single choice from $\{\tfrac12x + \tfrac12x' \mid x \in A, x' \in  A'\}\equiv \tfrac12A + \tfrac12A'$ (and likewise for $B$ and $B'$). Moreover, as shown in the bottom of Figure \ref{fig:normalconesex}, $\tfrac12A + \tfrac12A'$ and $\frac12B + \frac12B'$ yield the same identifiable sets.\footnote{For example, the set of utilities which find $\tfrac12(\frac12a + \tfrac12b) + \tfrac12(\frac{6}{10}b + \tfrac{4}{10}c)$ maximal from $\tfrac12A + \tfrac12A'$ is exactly those that find $\tfrac12b + \tfrac12(\frac23a + \tfrac13b)$ maximal from $\tfrac12B + \tfrac12B'$ (the set $W_1$).}
	
	\begin{figure}
		\centering
		\begin{minipage}{0.24\textwidth}
			\centering
			\begin{tikzpicture}[scale=.43]
				\begin{ternaryaxis}[
					title={\large $A = \{a, \tfrac12 a + \tfrac12 b,  \tfrac12 a + \tfrac12c, \tfrac12 b + \tfrac12c\}$},
					ternary limits relative=false,
					xmin=0,
					xmax=100,
					ymin=0,
					ymax=100,
					zmin=0,
					zmax=100, 
					xlabel={prob(a)},
					ylabel={prob(b)},
					zlabel={prob(c)},
					label style={sloped},
					minor tick num=3,
					grid=both,
					grid style={line width=0.3pt, opacity=.3}, 
					area style,
					tick label style={font=\footnotesize, opacity=.3},
					clip=false,
					disabledatascaling,
					]
					
					\addplot3+[opacity=.2, color=lam1, mark size=3pt,] table {
						x       y       z     
						100      0       0  
						50      50     0   
						0      50       50 
						50      0       50 
					};
					
					\addplot3+[only marks, color=lam1, mark size=3pt] table {
						x       y       z     
						100      0       0  
						50      50     0   
						0      50       50 
						50      0       50 
					};
					
				\end{ternaryaxis}
			\end{tikzpicture}
		\end{minipage}
		\hfill
		\begin{minipage}{0.24\textwidth}
			\centering
			\begin{tikzpicture}[scale=.43]
				\begin{ternaryaxis}[
					title={\large $A' = \{\tfrac{6}{10}b + \tfrac{4}{10}c,  \tfrac{4}{10}b + \tfrac{6}{10}c\}$},
					ternary limits relative=false,
					xmin=0,
					xmax=100,
					ymin=0,
					ymax=100,
					zmin=0,
					zmax=100, 
					xlabel={prob(a)},
					ylabel={prob(b)},
					zlabel={prob(c)},
					label style={sloped},
					minor tick num=3,
					grid=both,
					grid style={line width=0.3pt, opacity=.3}, 
					area style,
					tick label style={font=\footnotesize, opacity=.3},
					clip=false,
					disabledatascaling,
					]
					
					\draw[lam1, ultra thick, opacity=.5] (axis cs:.25, 40.125, 60.125) -- (axis cs:.25, 60.125, 40.125);
					
					\addplot3+[only marks, color=lam1, mark size=3pt,] table {
						x       y       z     
						0      40    60  
						0      60    40   
						
					};
					
				\end{ternaryaxis}
			\end{tikzpicture}
		\end{minipage}
		\hfill
		\begin{minipage}{0.24\textwidth}
			\centering
			\begin{tikzpicture}[scale=.43]
				\begin{ternaryaxis}[
					title={\large $B = \{a,b,c\}$},
					ternary limits relative=false,
					xmin=0,
					xmax=100,
					ymin=0,
					ymax=100,
					zmin=0,
					zmax=100, 
					xlabel={prob(a)},
					ylabel={prob(b)},
					zlabel={prob(c)},
					label style={sloped},
					minor tick num=3,
					grid=both,
					grid style={line width=0.3pt, opacity=.3}, 
					area style,
					tick label style={font=\footnotesize, opacity=.3},
					clip=false,
					disabledatascaling,
					]
					
					\addplot3+[opacity=.2, color=lam1, mark size=3pt,] table {
						x       y       z     
						100      0       0  
						0       100   0
						0       0     100
					};
					
					\addplot3+[only marks, color=lam1, mark size=3pt,] table {
						x       y       z     
						100      0       0  
						0       100   0
						0       0     100
					};
					
				\end{ternaryaxis}
			\end{tikzpicture}
		\end{minipage}
		\hfill
		\begin{minipage}{0.24\textwidth}
			\centering
			\begin{tikzpicture}[scale=.43]
				\begin{ternaryaxis}[
					title={\large $B' = \{\tfrac23 a + \tfrac13b, \tfrac23 a + \tfrac13c, \tfrac13 a + \tfrac13b + \tfrac13c\}$},
					ternary limits relative=false,
					xmin=0,
					xmax=100,
					ymin=0,
					ymax=100,
					zmin=0,
					zmax=100, 
					xlabel={prob(a)},
					ylabel={prob(b)},
					zlabel={prob(c)},
					label style={sloped},
					minor tick num=3,
					grid=both,
					grid style={line width=0.3pt, opacity=.3}, 
					area style,
					tick label style={font=\footnotesize, opacity=.3},
					clip=false,
					disabledatascaling,
					]
					
					\addplot3+[opacity=.2, color=lam1, mark size=3pt,] table {
						x       y       z 
						33.33      33.33    33.33  
						66.66   33.33    00 
						66.66  00    33.33 
					};
					
					\addplot3+[only marks, color=lam1, mark size=3pt,] table {
						x       y       z     
						33.33      33.33    33.33  
						66.66   33.33    00 
						66.66  00    33.33 
					};
					
				\end{ternaryaxis}
			\end{tikzpicture}
		\end{minipage}
		
		\vspace*{3ex}
		
		\begin{minipage}{0.3\textwidth}
			\centering
			\begin{tikzpicture}[scale=.6]
				\begin{ternaryaxis}[
					title={$\frac12A + \frac12A'$},
					ternary limits relative=false,
					xmin=0,
					xmax=100,
					ymin=0,
					ymax=100,
					zmin=0,
					zmax=100, 
					xlabel={\scriptsize prob(a)},
					ylabel={\scriptsize prob(b)},
					zlabel={\scriptsize prob(c)},
					label style={sloped},
					minor tick num=3,
					grid=both,
					grid style={line width=0.3pt, opacity=.3}, 
					area style,
					tick label style={font=\scriptsize, opacity=0},
					clip=false,
					disabledatascaling,
					]
					
					\addplot3+[opacity=.2, color=lam1] table {
						x       y       z 
						50     20     30  
						25     20     55
						00     45     55
						00     55     45
						25     55     20
						50    30     20
					};
					
					\addplot3+[only marks, color=blue, opacity=.4] table {
						x       y       z     
						25     45     30   
						25     30     45   
					};
					
					\addplot3+[only marks, color=lam1] table {
						x       y       z     
						50    30     20
						50     20     30  
						00     45     55
						00     55     45
						25     55     20
						25     20     55
					};
					
					\draw[lam2, dotted, opacity=.8] (axis cs: 50,20) -- (axis cs: 60,0);
					\draw[lam2, dotted, opacity=.8] (axis cs: 50,20) -- (axis cs: 70,10);
					\fill[lam2, dotted, opacity=.2] (axis cs: 50,20) -- (axis cs: 60,0) -- (axis cs: 70,10) -- cycle;
					\node at (axis cs: 60,10) [lam2, opacity=.8] {$W_3$};

					\draw[lam2, dotted, opacity=.8] (axis cs: 25,20) -- (axis cs: 35,0);
					\draw[lam2, dotted, opacity=.8] (axis cs: 25,20) -- (axis cs: 15,10);
					\fill[lam2, dotted, opacity=.2] (axis cs: 25,20) -- (axis cs: 35,0) -- (axis cs: 15,10) -- cycle;
					\node at (axis cs: 25,10) [lam2, opacity=.8] {$W_4$};
					
					\draw[lam2, dotted, opacity=.8] (axis cs: 0,45) -- (axis cs: -10,35);
					\draw[lam2, dotted, opacity=.8] (axis cs: 0,45) -- (axis cs: -20,55);
					\fill[lam2, dotted, opacity=.2] (axis cs: 0,45) -- (axis cs: -10,35) -- (axis cs: -20,55) -- cycle;
					\node at (axis cs: -10,45) [lam2, opacity=.8] {$W_5$};

					\draw[lam2, dotted, opacity=.8] (axis cs: 0,55) -- (axis cs: -10,75);
					\draw[lam2, dotted, opacity=.8] (axis cs: 0,55) -- (axis cs: -20,65);
					\fill[lam2, dotted, opacity=.2] (axis cs: 0,55) -- (axis cs: -10,75) -- (axis cs: -20,65) -- cycle;
					\node at (axis cs: -10,65) [lam2, opacity=.8] {$W_6$};
					
					\draw[lam2, dotted, opacity=.8] (axis cs: 25,55) -- (axis cs: 35,65);
					\draw[lam2, dotted, opacity=.8] (axis cs: 25,55) -- (axis cs: 15,75);
					\fill[lam2, dotted, opacity=.2] (axis cs: 25,55) -- (axis cs: 35,65) -- (axis cs: 15,75) -- cycle;
					\node at (axis cs: 25,65) [lam2, opacity=.8] {$W_1$};

					\draw[lam2, dotted, opacity=.8] (axis cs: 50,30) -- (axis cs: 60,40);
					\draw[lam2, dotted, opacity=.8] (axis cs: 50,30) -- (axis cs: 70,20);
					\fill[lam2, dotted, opacity=.2] (axis cs: 50,30) -- (axis cs: 60,40) -- (axis cs: 70,20) -- cycle;
					\node at (axis cs: 60,30) [lam2, opacity=.8] {$W_2$};

				\end{ternaryaxis}
			\end{tikzpicture}
		\end{minipage}
		\hfill
		\begin{minipage}{0.3\textwidth}
			\centering
			\begin{tikzpicture}[scale=.6]
				\begin{ternaryaxis}[
					title={$\frac12B + \frac12B'$},
					ternary limits relative=false,
					xmin=0,
					xmax=100,
					ymin=0,
					ymax=100,
					zmin=0,
					zmax=100, 
					xlabel={\scriptsize prob(a)},
					ylabel={\scriptsize prob(b)},
					zlabel={\scriptsize prob(c)},
					label style={sloped},
					minor tick num=3,
					grid=both,
					grid style={line width=0.3pt, opacity=.3}, 
					area style,
					tick label style={font=\scriptsize, opacity=0},
					clip=false,
					disabledatascaling,
					]
					
					\addplot3+[opacity=.2, color=lam1] table {
						x       y       z 
						83.33    00     16.66
						83.33    16.66     00
						33.33    66.66     00
						16.66    66.66    16.66
						16.66    16.66    66.66
						33.33     00    66.66
						
					};
					
					\addplot3+[only marks, color=blue, opacity=.4] table {
						x       y       z     
						66.66    16.66    16.66
						33.33    50     16.66
						33.33    16.66      50
					};
					
					\addplot3+[only marks, color=lam1] table {
						x       y       z     
						83.33    16.66     00
						83.33    00     16.66
						33.33    66.66     00
						16.66    66.66    16.66
						33.33     00    66.66
						16.66    16.66    66.66
					};
					
					\draw[lam2, dotted, opacity=.8] (axis cs: 83.33,0) -- (axis cs: 93.33,-20);
					\draw[lam2, dotted, opacity=.8] (axis cs: 83.33,0) -- (axis cs: 103.33,-10);
					\fill[lam2, dotted, opacity=.2] (axis cs: 83.33,0) -- (axis cs: 93.33,-20) -- (axis cs: 103.33,-10) -- cycle;
					\node at (axis cs: 93.33,-10) [lam2, opacity=.8] {$W_3$};
					
					\draw[lam2, dotted, opacity=.8] (axis cs: 33.33,0) -- (axis cs: 43.33,-20);
					\draw[lam2, dotted, opacity=.8] (axis cs: 33.33,0) -- (axis cs: 23.33,-10);
					\fill[lam2, dotted, opacity=.2] (axis cs: 33.33,0) -- (axis cs: 43.33,-20) -- (axis cs: 23.33,-10) -- cycle;
					\node at (axis cs: 33.33,-10) [lam2, opacity=.8] {$W_4$};
					
					\draw[lam2, dotted, opacity=.8] (axis cs: 16.66,16.66) -- (axis cs: 06.66,06.66);
					\draw[lam2, dotted, opacity=.8] (axis cs: 16.66,16.66) -- (axis cs: -3.66,26.66);
					\fill[lam2, dotted, opacity=.2] (axis cs: 16.66,16.66) -- (axis cs: 6.66,6.66) -- (axis cs: -3.66,26.66) -- cycle;
					\node at (axis cs: 6.55,16.66) [lam2, opacity=.8] {$W_5$};
					
					\draw[lam2, dotted, opacity=.8] (axis cs: 16.66,66.66) -- (axis cs: 06.66,86.66);
					\draw[lam2, dotted, opacity=.8] (axis cs: 16.66,66.66) -- (axis cs: -03.66,76.66);
					\fill[lam2, dotted, opacity=.2] (axis cs: 16.66,66.66) -- (axis cs: 6.66,86.66) -- (axis cs: -03.66,76.66) -- cycle;
					\node at (axis cs: 6.55,76.66) [lam2, opacity=.8] {$W_6$};
					
					\draw[lam2, dotted, opacity=.8] (axis cs: 33.33,66.66) -- (axis cs: 43.33,76.66);
					\draw[lam2, dotted, opacity=.8] (axis cs: 33.33,66.66) -- (axis cs: 23.66,86.66);
					\fill[lam2, dotted, opacity=.2] (axis cs: 33.33,66.66) -- (axis cs: 43.33,76.66) -- (axis cs: 23.66,86.66) -- cycle;
					\node at (axis cs: 33.44,76.66) [lam2, opacity=.8] {$W_1$};
					
					\draw[lam2, dotted, opacity=.8] (axis cs: 83.33,16.66) -- (axis cs: 93.33,26.66);
					\draw[lam2, dotted, opacity=.8] (axis cs: 83.33,16.66) -- (axis cs: 103.33,6.66);
					\fill[lam2, dotted, opacity=.2] (axis cs: 83.33,16.66) -- (axis cs: 93.33,26.66) -- (axis cs: 103.33,6.66) -- cycle;
					\node at (axis cs: 93.33,16.66) [lam2, opacity=.8] {$W_2$};

				\end{ternaryaxis}
			\end{tikzpicture}
		\end{minipage}
		\hfill
		\begin{minipage}{0.3\textwidth}
			\centering
			\begin{tikzpicture}[scale=.6]
				\begin{axis}[
					title = {$\U \cong \R^2$},
					axis lines=middle, 
					axis equal,        
					xmin=-10, xmax=10, 
					ymin=-10, ymax=10, 
					xlabel=\empty,
					ylabel=\empty,
					xticklabel=\empty,
					yticklabel=\empty,
					grid=both,
					enlargelimits={abs=0.5}, 
					domain=-10:10,
					samples=100,
					minor tick num=5,
					grid=both,
					grid style={line width=0.3pt, opacity=.3}, 
					area style,
					axis line style={opacity=.1}, 
					tick style={opacity=0},
					]
					
					\addplot[thick, lam2, dotted,] coordinates {(0,-10) (0,10)};
					
					\addplot[thick, dotted, lam2, domain=-10:10] ({x}, {x * tan(30)});
					
					\addplot[thick, lam2, dotted, domain=-10:10] ({x}, {x * tan(-30)});
					
					\addplot[only marks, color=lam2] (0,0);
					\node at (axis cs: 3,5) [lam2, opacity=.8] {$W_3$};
					\node at (axis cs: 5,0) [lam2, opacity=.8] {$W_4$};
					\node at (axis cs: 3,-5) [lam2, opacity=.8] {$W_5$};
					\node at (axis cs: -3,-5) [lam2, opacity=.8] {$W_6$};
					\node at (axis cs: -5,0) [lam2, opacity=.8] {$W_1$};
					\node at (axis cs: -3,5) [lam2, opacity=.8] {$W_2$};
					
				\end{axis}
			\end{tikzpicture}
		\end{minipage}
		\caption{\emph{Top}: The four decision problems in experiments \rexp{a} and \rexp{b} represented in the simplex. The convex hull of the decision problems is shaded.  \emph{Bottom}: The identifiable sets from $\tfrac12A + \tfrac12A'$ and $\tfrac12B + \tfrac12B'$. These form the same partition of $\U$ as shown in the third panel.}
		\label{fig:normalconesex}
	\end{figure}
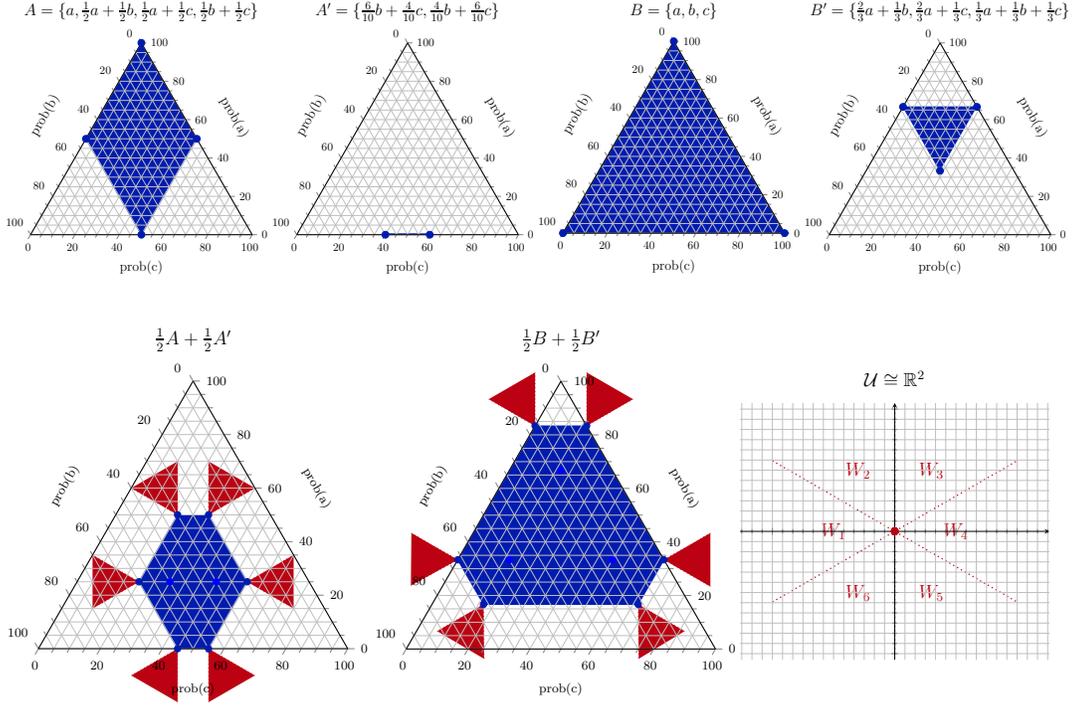
	
	Within the domain of linear utility, Structural Invariance is captured by a translation invariance axiom, which arises from the specific characteristics of the environment. Conceptually, Structural Invariance states that changing aspects of an experiment that do not affect what it can identify should not change its value; linear translations are the specific structural property invariant for linear utilities. 
	As such, applying our results for expected utility (Section \ref{sec:EUmodel}) to non-linear models (e.g. ambiguity averse utility functions) only requires identifying the appropriate invariance axiom.
	
	To understand Identification Separability, consider the following four partitions of the decision problem $A$ (from the earlier example); these partitions are shown in Figure \ref{fig:partionsIdenSep}.
	\begin{align*}
		\P\phantom{'} &= \big\{ \{a\}, \ \{\tfrac12 a + \tfrac12 b\}, \ \{ \tfrac12 a + \tfrac12c \}, \{ \tfrac12 b + \tfrac12c\} \big\} \\
		\P' &= \big\{ \{a, \ \tfrac12 a + \tfrac12 b\}, \ \{ \tfrac12 a + \tfrac12c ,  \tfrac12 b + \tfrac12c\} \big\} \\
		\Q\phantom{'} &= \big\{ \{a, \ \tfrac12 a + \tfrac12 b\}, \ \{ \tfrac12 a + \tfrac12c \}, \{ \tfrac12 b + \tfrac12c\} \big\} \\
		\Q' &= \big\{ \{a\}, \ \{\tfrac12 a + \tfrac12 b\}, \ \{ \tfrac12 a + \tfrac12c , \ \tfrac12 b + \tfrac12c\} \big\} \\
	\end{align*}
	Based on these decision problems, the analyst considers two randomized experiments:
	\begin{itemize}[topsep=5pt,itemsep=-.2em,leftmargin=12ex]
		\item[\textsc{exp c}\ :]\label{exp:c}  The analyst offers $(A,\P)$ and  $(A,\P')$ each with probability with $\frac12$.
		\item[\textsc{exp d}\ :]\label{exp:d} The analyst offers $(A,\Q)$ and  $(A,\Q')$ each with probability with $\frac12$.
	\end{itemize}

	\begin{figure}
		\centering
		\begin{minipage}{0.24\textwidth}
			\centering
			\begin{tikzpicture}[scale=.43]
				\begin{ternaryaxis}[
					title={\Large $\P$},
					ternary limits relative=false,
					xmin=0,
					xmax=100,
					ymin=0,
					ymax=100,
					zmin=0,
					zmax=100, 
					xlabel={prob(a)},
					ylabel={prob(b)},
					zlabel={prob(c)},
					label style={sloped},
					minor tick num=3,
					grid=both,
					grid style={line width=0.3pt, opacity=.3}, 
					area style,
					tick label style={font=\footnotesize, opacity=.3},
					clip=false,
					disabledatascaling,
					]
					
					\addplot3+[opacity=.2, color=lam1, mark size=3pt,] table {
						x       y       z     
						100      0       0  
						50      50     0   
						0      50       50 
						50      0       50 
					};
					
					\addplot3+[only marks, color=lam1, mark size=3pt] table {
						x       y       z     
						100      0       0  
						50      50     0   
						0      50       50 
						50      0       50 
					};
					
				\end{ternaryaxis}
			\end{tikzpicture}
		\end{minipage}
		\hfill\begin{minipage}{0.24\textwidth}
			\centering
			\begin{tikzpicture}[scale=.43]
				\begin{ternaryaxis}[
					title={\Large $\P'$},
					ternary limits relative=false,
					xmin=0,
					xmax=100,
					ymin=0,
					ymax=100,
					zmin=0,
					zmax=100, 
					xlabel={prob(a)},
					ylabel={prob(b)},
					zlabel={prob(c)},
					label style={sloped},
					minor tick num=3,
					grid=both,
					grid style={line width=0.3pt, opacity=.3}, 
					area style,
					tick label style={font=\footnotesize, opacity=.3},
					clip=false,
					disabledatascaling,
					]
					
					\draw [rotate=60, red, ultra thick, dotted] (axis cs: 0,25) ellipse (2.2cm and .4cm); 
					\draw [rotate=60, red, ultra thick, dotted] (axis cs: -50,75) ellipse (2.2cm and .4cm);

					\addplot3+[opacity=.2, color=lam1, mark size=3pt,] table {
						x       y       z     
						100      0       0  
						50      50     0   
						0      50       50 
						50      0       50 
					};
					
					\addplot3+[only marks, color=lam1, mark size=3pt] table {
						x       y       z     
						100      0       0  
						50      50     0   
						0      50       50 
						50      0       50 
					};
					
				\end{ternaryaxis}
			\end{tikzpicture}
		\end{minipage}
		\hfill\begin{minipage}{0.24\textwidth}
			\centering
			\begin{tikzpicture}[scale=.43]
				\begin{ternaryaxis}[
					title={\Large $\Q$},
					ternary limits relative=false,
					xmin=0,
					xmax=100,
					ymin=0,
					ymax=100,
					zmin=0,
					zmax=100, 
					xlabel={prob(a)},
					ylabel={prob(b)},
					zlabel={prob(c)},
					label style={sloped},
					minor tick num=3,
					grid=both,
					grid style={line width=0.3pt, opacity=.3}, 
					area style,
					tick label style={font=\footnotesize, opacity=.3},
					clip=false,
					disabledatascaling,
					]
					
					\draw [rotate=60, red, ultra thick, dotted] (axis cs: 0,25) ellipse (2.2cm and .4cm); 
					
					\addplot3+[opacity=.2, color=lam1, mark size=3pt,] table {
						x       y       z     
						100      0       0  
						50      50     0   
						0      50       50 
						50      0       50 
					};
					
					\addplot3+[only marks, color=lam1, mark size=3pt] table {
						x       y       z     
						100      0       0  
						50      50     0   
						0      50       50 
						50      0       50 
					};
					
				\end{ternaryaxis}
			\end{tikzpicture}
		\end{minipage}
		\hfill\begin{minipage}{0.24\textwidth}
			\centering
			\begin{tikzpicture}[scale=.43]
				\begin{ternaryaxis}[
					title={\Large $\Q '$},
					ternary limits relative=false,
					xmin=0,
					xmax=100,
					ymin=0,
					ymax=100,
					zmin=0,
					zmax=100, 
					xlabel={prob(a)},
					ylabel={prob(b)},
					zlabel={prob(c)},
					label style={sloped},
					minor tick num=3,
					grid=both,
					grid style={line width=0.3pt, opacity=.3}, 
					area style,
					tick label style={font=\footnotesize, opacity=.3},
					clip=false,
					disabledatascaling,
					]
					
					\draw [rotate=60, red, ultra thick, dotted] (axis cs: -50,75) ellipse (2.2cm and .4cm); 
					
					\addplot3+[opacity=.2, color=lam1, mark size=3pt,] table {
						x       y       z     
						100      0       0  
						50      50     0   
						0      50       50 
						50      0       50 
					};
					
					\addplot3+[only marks, color=lam1, mark size=3pt] table {
						x       y       z     
						100      0       0  
						50      50     0   
						0      50       50 
						50      0       50 
					};
					
				\end{ternaryaxis}
			\end{tikzpicture}
		\end{minipage}

		\caption{The four partitions of $A$ from \rexp{c} and \rexp{d}.}
		\label{fig:partionsIdenSep}
	\end{figure}
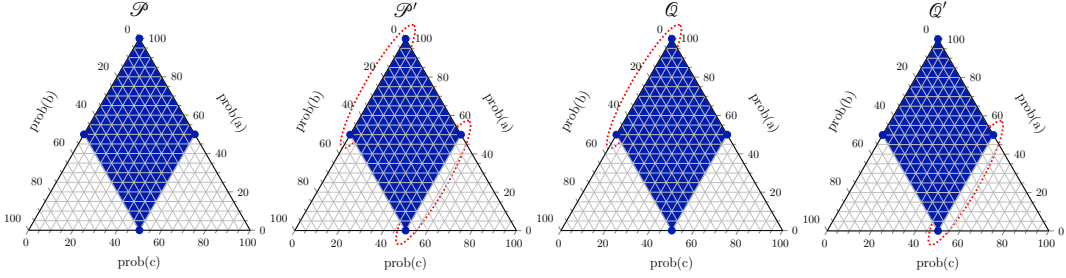

	As in the previous example, what might seem to depend on the aims of the analyst is in fact dictated by criteria of rational design; the principle of Identification Separability requires these two experiments are valued equally. To see this, notice that under \rexp{c}, with probability $\frac12$ the analyst learns exactly which element was chosen, and with the remaining probability $\frac12$, she learns only which cell of $\P'$ contained the chosen alternative. Although less immediate, this is also the case for \rexp{d}. Indeed, let $W \subset \U$ denote the set of utilities such that
	the subject will choose an element of the first cell of $\P'$ (i.e., will choose either $a$ or $\tfrac12 a + \tfrac12 b$) and $W^c \subset \U$ those utilities such that the subject will choose an element of the second (i.e., either $\tfrac12 a + \tfrac12c$ or $\tfrac12 b + \tfrac12c$).\footnote{Using the notation from the bottom of Figure \ref{fig:normalconesex}, $W = W_1 \cup W_2 \cup W_3$ and $W^c = W_4 \cup W_4 \cup W_5$; notice we are excluding the possibility of the subject being indifferent between alternatives. Following the literature on random utility, we assume such ties occur with zero probability. See section \ref{sec:genmodel}.} Then, conditional on $u^* \in W$, the subject will choose an element out of the first cell of $\P'$: under $\Q$ this is all that is observed, while under $\Q'$ the choice is observed perfectly. Conditional on $u^* \in W^c$, the same logic applies: $\Q$ perfectly reveals the subject's choice, while $\Q'$ only that the second cell of $\P'$ contains the chosen alternative.
	
	Thus, both \rexp{c} and \rexp{d} reveal the subject's choice half of the time and the cell of $\P'$ containing his choice the other half. The difference between these two experiments is that in the latter, the amount of information revealed is correlated with the the subject's utility type. That is, if \rexp{d} ends up perfectly revealing the subject's choice when $u^* \in W$, we know that it \emph{would not have done} had  $u^* \in W^c$, and vice versa. The principle of Identification Separability dictates that such counterfactual assessments are irrelevant, and thus, that the two randomized experiments are valued equally.

	
	\hd{Functional Forms} The expected identification value representation is flexible enough to accommodate many Bayesian theories of optimal experimental design. 
	For instance, by taking
	$$\tau(W)=-\log(\mu(W)),$$
	the value of an experiment is its expected reduction in entropy relative to the prior \citep{cover1991entropy}. We axiomatize this special case in Section \ref{sec:entropy}. \cite{drake2022bayesian} propose a dynamic Bayesian procedure for preference identification on the basis of this functional form.
	Another special case of our index comes from hypothesis testing. An analyst who wishes to test if the subjects preference is in some set $W^*$ would consider 
	\begin{equation*}
		\tau(W)  = 
		\begin{dcases}
			1 &\text{ if } W \subseteq W^* \text{ or } W^* \subseteq W^c  \\
			0 &\text{ otherwise }.
		\end{dcases}
	\end{equation*}
	
	Finally, we observe that although our theory is meant to contribute to the conceptual understanding of choice experiments, it is flexible enough to allow for functional forms unrelated to experimental design. For instance, a Bayesian principal may only want to promote agents with similar preferences to herself. Hence, she conducts a test to see what kind of preferences her agents have when making promotion decisions. The following specification allows for such interpretation
	$$\tau(W)=\max_{a\in\{0,1\}}\integral_{W}\xi(a,u)d\mu.$$
	where $a=1$ ($a=0$) is interpreted as (not) promoting the agent and $\xi(a,u)$ is her utility of promoting an agent that has preference $u$ and $\mu$ is her prior over the agent's preference.
	
	\hd{Outline} The paper proceeds as follows: The introduction concludes with a review of the relevant literature. The model is presented in Section \ref{sec:genmodel}. Our normative principles and main representation result are in Section \ref{sec:axioms}. 
	Sections \ref{sec:EUmodel} and \ref{sec:entropy} discuss the special cases of Expected Utility maximizing subjects and entropy reduction, respectively.
	Section \ref{sec:belief-free} outlines a version of the model without prior beliefs.
	Finally, Section \ref{sec:partial-observability} concludes by showing how our framework is general enough to capture dynamic experiments. All proofs are collected in the Appendix.
	
	\subsection{Related Literature}
	This paper joins the large literature in economics on eliciting preferences from observable behavior.
	It differs from most of the literature as it does not focus on efficiency of a particular elicitation method, but on what are the minimal properties a method should satisfy in order to be considered rational. 
	Such questions have been suggested in the statistics literature on Bayesian experimental design. Early texts such as \cite{raiffa2000applied} and \cite{lindley1972bayesian} propose a utility function for Bayesian experimenters. 
	The literature that followed provided specializations of the general function to feet more stylized settings such as regression analysis and model discrimination analysis (see \cite{chaloner1995bayesian} for a review of the literature). Thus, our work complements the existing literature by providing a framework to discuss experimental design as well as the missing axiomatic foundation.      
	
	Our framework is inspired by the literature in economics on Discrete Choice Experiments (henceforth DCE). 
	DCE's were initially developed by \cite{louviere1982design}. They are based on the theoretical framework of the Random Utility Model (\cite{luce1959possible} and \cite{mcfadden1973precedence}). 
	We contribute to the DCE literature by providing a unifying framework to analyse deviations from the standard DCE method. Given the recent interest in employing dynamic procedures to substitute DCE's, our results can be used as a test for such procedures. If they do not satisfy our axioms, they should not be employed.
	
	Outside of the DCE literature but within the random utility literature, \cite{gul2006random} (henceforth GP) provide necessary and sufficient conditions for random choice data to be consistent with \emph{random} expected utility. We use their work as a building block in providing foundations to Bayesian procedures. Specifically, our richness condition described in Section \ref{sec:EUmodel} are direct consequences of the GP assumptions.
	
	\cite{gilboa1991value} studies a related problem to ours. Their goal is to provide axiomatic foundations for functions over partitions of states that can be interpreted as describing the value of information for some Bayesian agent. Our analysis can also be interpreted as providing foundations for functions that can be interpreted as describing the value of additional information of an agent's preference. There are two key differences. First, we take as observable preference over experiments as opposed to a function over partitions of the utility space, which would be the analog of their domain in our setting. Second, we do not look for identification functions for which there exists a Bayesian experimenter that may employ them. Indeed, we do not take rationality as given and look for identification functions that satisfy it. We propose a notion of rationality and characterize the set of identification functions that satisfy it.
	
	Finally, our work also contributes to the preference over menus in decision theory. Starting with \cite{gul2001temptation} and \cite{dekel2001representing}, economists have used preferences over menus of lotteries to study distinct phychological phenomena such as temptation and self-control. The literature then generalized the domain to lotteries over menus of lotteries to obtain sharper results (\cite{stovall2018temptation} and \cite{ergin2015hidden}). Our work shows that lotteries over menus can also be employed to analyze experimental design. Thus, it suggests that some of the earlier work in decision theory may be useful for experimental design.
	
	\section{General Model} \label{sec:genmodel}
	
	\subsection{Preliminaries} 
	
	An \emph{abstract experimental environment} is a tuple $(Z,\U, \W, \mu)$ where $Z$ is some set of possible choice alternatives, $\U \subseteq \{u: Z \to R\}$ a set of utility types, $\W$ is an algebra of measurable sets of $\U$ and $\mu$ probability distribution over $(\U,\W)$. We interpret $(Z,\U, \W, \mu)$ as the theory the analyst has about the subject's preferences.
	A \emph{decision problem} $A$ is a finite subset of $Z$. Let $\D$ denote the set of all decision problems. 
	
	Given a decision problem $A$ and some $B \subseteq A$, let
	$$W_{A,B} = \{u \in \U, B \cap \argmax_{x \in A} u(x) \neq \emptyset\}$$
	denote the set of utilities for which some element of $B$ is a maximizer when facing the decision problem $A$. Intuitively, $W_{A,B}$ is the set of preferences that would choose an element in $B$ when facing menu $A$. 
	
	To achieve her goal, the analyst can can offer the subject a decision problem from which she will observe the subject's choice. 
	While it is plausible that a subject's behavior can be observed perfectly in static laboratory conditions, in dynamic settings or more general environments (i.e., field experiments, consumer testing in industry, data collection by online platforms, etc.), the analyst may only be able to partially observe choice. To allow for such constraints, we define an experiment as a decision problem and a partition.\footnote{A partition $\P$ of $X$ is a set of subsets of $X$ that are mutually disjoint and whose union is $X$.} The interpretation is that $P \in \P$ represents what the analyst observes when the subject's choice out of $A$ is contained in $P$. 
	
	Formally, an \emph{experiment} $e = (A,\P)$ is a pair where where $A \in \D$ and $\P$ is a partition of $A$ such that for any $P,Q \in \P$
	\begin{enumerate}[label=(\texttt{E\arabic*}), itemsep=-.4em]
		\item\label{exp:measure} $W_{A,P} \in \W$
		\item\label{exp:nooverlap} $\mu(W_{A,P} \cap W_{A,Q}) = 0$
	\end{enumerate}
	The first requirement states that analyst assigns a prior probability to each observable outcome, and the second states that the analyst can unambiguously interpret the observed outcome. Specifically, the analyst places $\mu$-probability 0 on the subject being indifferent between two alternatives in $A$ so that observed choice can be interpreted without worrying about how ties are broken.
	
	Our notion of experiments can be used to define partial identification: A set of preferences $W\subseteq \U$ is \emph{identifiable} in $(A,\P)$ if $W=\UAP$ for some $P\in \P$. Given an experiment $(A,\P)$, the analyst can calculate the family $\{\UAP|P\in \P\}$, the sets of preferences that are identifiable by the experiment. 
	
	Call two (finite) collections of subsets of $\U$, $\{W_1, \ldots W_n\}$ and $\{V_1, \ldots W_m\}$ \emph{$\mu$-equivalent} if for every $W_i$ with $\mu(W_i) > 0$ there exists a $V_j$ such that $\mu(W_i) = \mu(W_i \cap V_j) = \mu(V_j)$, and vice versa. That is, the collections are $\mu$-equivalent if they are equal up to measure 0 sets.
	
	We assume the analyst has access to a set of experiments $\E$ that satisfies the following two properties:
	\begin{itemize}
		\item $(A,\P) \in \E$ and $\Q$ is a coarsening of $\P$ then $(A,\Q) \in \E$,
		\item For any finite $\W$-measurable partition $\mathcal{W}$ of $\U$, there is a some experiment $(A,\P) \in \E$ such that $\{\UAP|P\in \P\}$ is $\mu$-equivalent to $\mathcal{W}$.
	\end{itemize}
	The first property states that if $(A,\P)$ is feasible, then an experiment that potentially identifies less utilities is also feasible. The second property demands that for any finite partition of the utility space, the analyst can always find an experiment that would induce such a partition. We call such sets of experiments \emph{rich}.
	

	Given a set of experiments $\E$, a \emph{randomized experiment} (over $\E$) $\pi$ is a finitely supported probability distribution over $\E$. The set of all randomized experiments over $\E$ is denoted by $\Pi(\E)$. For a given randomized experiment, $\pi$, let $\supp(\pi) = \{e \in \E \mid \pi(e) > 0\}$ denote the support of $\pi$. Our primitive is the analyst's preference, $\s$, over the set of all randomized experiments over some rich set of experiments.  
	
	\subsection{Representation}
	
	A \emph{expected identification value representation} for $\s$ is the following:
	\begin{equation}
		\label{eq:rep2}
		\tag{$\star\star$}
		F(\pi)=\sum_{\supp(\pi)}\Big( \sum_{P \in \P} \tau(W_{A,P}) \mu(W_{A,P}) \Big) \pi(A,\P)
	\end{equation}
	where $\tau: \Omega \to \R$ satisfies
	\begin{enumerate}[label=(\texttt{T\arabic*}), itemsep=-.4em]
		\item\label{tau:subadd} For all non-$\mu$-null $V$ and $W \subseteq V$, 
		$$\tau(W)\mu(W\, |\, V) + \tau(V \setminus W)(1-\mu(W\, |\, V)) \geq \tau(V),$$
		with equality holding whenever $\mu(W) = 0$.
		\item\label{tau:empty}$\tau(\U) = 0$. 
	\end{enumerate}
	
	Condition \ref{tau:subadd} states that information is never bad for the analyst. Indeed, consider an analyst who has already made the identification $V \subseteq \U$---that is, who already knows that the subject's preference is contained in $V$---and is contemplating the value of an additional observation that would reveal if the subject's preference is in $W$. The current value of her identification is $\tau(V)$. If she learns the additional observation, the total value will depend on if the subject's preference lies in $W$ or not, resulting in $\tau(W)$ or $\tau(V \setminus W)$ respectively. According to her beliefs, the former occurs with probability $\mu(W\, |\, V)$ and the latter with probability $1-\mu(W\, |\, V)$. Thus, Condition \ref{tau:subadd} requires the expected value of this further information is (weakly) positive. Notice that if $\tau(W) \geq \tau(V)$ whenever $W \subseteq V$, then the constraint follows immediately. 
	
	In many cases, the analyst may not entertain a prior over $\U$. Nonetheless, our theory applies almost exactly. In this case, the value function $F$ can be written as $$F(A,\P)=\sum_{P \in \P} \nu(W_{A,P}),$$ for an abstract identification index $\nu$ which does not separate the intrinsic value of identification from its likelihood. This is akin to the failure of separation into tastes and beliefs in state-dependent expected utility. In this case, $\nu$ must be sub-additive to imbue a positive value for information.
	
	\section{Normative Principles of Experimental Design}
	\label{sec:axioms}

	If $\P$ is a partition of some set $X$ and $Y \subseteq X$, then $\P|_Y = \{P \cap Y \mid P \in \P\}$ is a partition of $Y$; we denote the corresponding (possibly empty) cells as $P|_Y$. 
	If $\P$ and $\Q$ are both partitions of the same set $X$ and $Y \subseteq X$ is measurable with respect to both $\P$ and $\Q$ then $\P_Y\Q$ denotes the partition that coincides with $\P$ over $Y$ and with $\Q$ over $X \setminus Y$.
	
	We impose four axioms on $\s$, the first of which requires that it admits an expected utility representation. This axiom is not expressed in terms of its individual choices, as its behavioral foundations are widely known.
	
	\begin{ax}{A}{Expected Utility}{eu}
		$\s$ entertains an expected utility representation.
	\end{ax}
	
	The following three axioms reflect our normative principles. Recall that our first principle, Information Monotonicity, asserts that the analyst has a preference for sharper identification. In the current domain, this amounts to assuming finer partitions will always be weakly preferred..

	\begin{ax}{A}{Monotonicity}{mon}
		For $A \in \D$, and partitions $\P$, $\Q$ of $A$ it follows that
		$$(A, \P) \s (A,\Q)$$ whenever $\P$ is finer than $\Q$.
	\end{ax}
	
	Our second principle maintains that the value of an experiment should only depend on what is potentially identifiable. Thus, it requires indifference between two experiments that have the same ex-ante identifiable set's of preferences (up to $\mu$-probability 0 events).
	
	\begin{ax}{A}{Structural Invariance}{si}
		Let $(A,\P)$ and $(B,\Q)$ be such that $\{\UAP|P\in \P\}$ is $\mu$-equivalent to $\{W_{B,Q} |Q \in \Q\}$.
		Then $(A,\P) \sim (B,\Q)$.
	\end{ax}
	
	
	Finally, Identification Separability demands that the value of some partial identification cannot depend on the counterfactual. We take advantage of our lottery domain to capture this. Specifically, fix a decision problem $A$ and partitions $\P$ and $\Q$ of $A$. 
	Identification Separability requires that for any subset $B \subseteq A$, the value of identification given $(A,\P)$ and $(A,\Q)$, conditional that the choice is in $B$,  should only depend on $\P|_{B}$ and $\Q|_{B}$, respectively. 
	Hence, if the agent will choose an element of $B$, a radomized experiment between  $(A,\P)$ and $(A,\Q)$ should be indifferent to a randomized experiment between  $(A, \P_{B}\Q)$ and $(A,\Q_{B}\P)$.
	
	\begin{ax}{A}{Identification Separability}{is}
		For $A \in \D$, partitions $\P$, $\Q$ of $A$, and $B\subseteq A$ measurable with respect to both $\P$ and $\Q$
		$$\tfrac12(A, \P) + \tfrac12(A,\Q) \sim \tfrac12(A, \P_{B}\Q) + \tfrac12(A,\Q_{B}\P).$$
	\end{ax}
	
	Requiring that the value of an object that is uncertain does not depend on the counterfactual is a well known implication of Dynamic Consistency and Consequentialism. As we now elaborate, our axiom is a direct implication of these requirements.  
	
	Consider an extension of the analyst's preferences to the case in which she knows choice out of $A$ will be contained in $B$, denoted $\s_B$, and the case in which she knows the opposite, denoted $\s_{B^c}$. If the choice is in $B$, then in terms of preference identification, the experiment $(A,\P)$ is equivalent to $(A,\P_B\Q)$ and $(A,\Q)$ to $(A,\Q_B\P)$. Analogously, if the choice is not in $B$, $(A,\Q)$ is equivalent to $(A,\P_B\Q)$ and $(A,\P)$ to $(A,\Q_B\P)$. Hence, if the analyst's preference do not depend on the counterfactual (consequentialism), then 
	\begin{align*}
		(A,\P) \sim_B (A,\P_B\Q)&\quad \mbox{ and }\quad (A,\Q)\sim_B(A,\Q_B\P); \\
		(A,\Q)\sim_{B^c}(A,\P_B\Q)&\quad \mbox{ and } \quad (A,\P)\sim_{B^c}(A,\Q_B\P).
	\end{align*}
	Therefore, under Independence, 
	\begin{align*}
		\tfrac12(A,\P) +\tfrac12(A,\Q) &\sim_B \tfrac12(A,\P_B\Q)+\tfrac12(A,\Q_B\P) \\
		\tfrac12(A,\P) +\tfrac12(A,\Q)&\sim_{B^c} \tfrac12(A,\P_B\Q)+\tfrac12(A,\Q_B\P).
	\end{align*}
	Finally, observe that if the analyst's ex-ante preference respects her conditional preferences (dynamic consistency), she must exhibit
	$$\tfrac12(A, \P) + \tfrac12(A,\Q) \sim \tfrac12(A, \P_{B}\Q) + \tfrac12(A,\Q_{B}\P).$$
	
	
	These four axioms---\r{ax:eu} providing the expected utility structure, and \r{ax:mon}--\r{ax:is} capturing our three normative principles---characterize expected identification value maximization.
	
	\begin{theorem}
		\label{thm:rep}
		The preference $\s$ satisfies \r{ax:eu}--\r{ax:is} if and only if it has an expected identification value representation.
	\end{theorem}
	
	\section{Identifying Expected Utility Preferences} \label{sec:EUmodel}
	
	The structural invariance axiom, \r{ax:si}, states abstractly that the value of an experiment should not depend on structural details. When the choice environment has a specific structure, this principle can be made concrete so as to reflect the particular invariant quantities of the environment at hand. We will now show how structural invariance captures tangible restrictions on the ranking of experiments within the specific environment of linear utility. Here, the analyst is interested in identifying the Von Neumann–Morgenstern utility index of the subject, under the maintained assumption that he is an expected utility maximizer. 
	
	This environment is closely related to the setup of random expected utility models \`a la \cite{gul2006random}. In particular, the experimenter's prior $\mu$ defines a GP random expected utility model.  Our conditions on experiments \ref{exp:measure} and \ref{exp:nooverlap} and our richness condition are then direct consequences of the GP assumptions.
	
	
	
	Let $\Delta$ be a convex and compact subset of a finite dimensional Euclidean space $\R^\ell$ and $\U^\Delta$ denote the set of expected utility (i.e., affine) functions over $\Delta$. So the set of decision problems $\D$ is the set of all finite subsets of $\Delta$. Let $\Omega^\Delta$ be the smallest algebra on $\U^\Delta$ that contains all identifiable sets: that is contains $W_{A,B}$ for all $A \in \D$ and $B \subseteq A$. Following  GP call a $\mu \in \Pr(\U^\Delta, \W^\Delta)$ \emph{regular} if $\mu(u\in\U^\Delta|u(x)=u(y))=0$ for all $x,y\in \Delta$.\footnote{GP shows that regular measures are exactly the measures that can be potentially identified from (random) choice data.} 
	
	Theorem \ref{thm:rich}, below, shows that our richness assumption is not overly strong; within the expected utility model, it is a natural consequence of standard assumptions over the ex-ante distribution on utilities.
	
	\begin{theorem}
		\label{thm:rich}
		If $\mu$ is regular, then $\E^\Delta = \{(A,\P) \mid A \subseteq \Delta \text{ is finite}, \P \text{ partitions } A\}$ is a rich set of experiments.
	\end{theorem}

	We will now recast structural invariance in a domain specific manner, illuminating the concrete notion of invariance that is relevant to the expected utility model. Specifically, we will show that structural invariance is equivalent to two axioms that specify when two experiments are equivalent and that do not need to directly reference sets of identifiable utilities. To do this, we first need to define a notion of mixing: For $A,B \subseteq \Delta$, and $\alpha \in [0,1]$, let $\alpha A + (1-\alpha) B = \{\alpha x + (1-\alpha) y \mid x \in A, y \in B\}$ denote the the Minkowski sum. If $A$ and $B$ are decision problems (i.e., are finite), then so is  $\alpha A + (1-\alpha) B$.
	
	Observe that under the assumption that $u$ is linear, if $x$ maximizes $u$ over $A$ and $y$ maximizes $u$ over $B$, then the mixture of $x$ and $y$ will maximize $u$ over the corresponding mixture of the menus. That is:
	\begin{equation*}
		\begin{rcases*}
			x \in \argmax_{A} u(\cdot) \\
			y \in \argmax_{B} u(\cdot)
		\end{rcases*} 
		\quad \text{ if and only if } \quad
		\alpha x + (1-\alpha)y \in \argmax_{\alpha A + (1-\alpha)B} u(\cdot)
	\end{equation*}
	for any $\alpha>0$ and $A,B \in \D$. Now consider the experiment $(A,\{P_1, \ldots P_n\})$ and some other decision problem $B$; by the above logic $u \in W_{A,P_i}$ if and only $u \in W_{\alpha A + (1-\alpha)B, \alpha P_i + (1-\alpha)B}$. Indeed, there must be some $y \in B$ that maximizes $u$ over $B$, so that $\alpha x + (1-\alpha)y \in \alpha P_i + (1-\alpha)B$ maximizes $u$ over the mixture.  Hence, translating an experiment by mixing both the decision problem and the observability partition with some common $B$ does not alter which preference sets can be identified. This particular form of invariance is captured by the following axiom.\footnote{The reason there is a subset, rather than set equality, in the axiom is that it is possible that $z \in \alpha A + (1-\alpha)B$ is not a unique mixture of two elements. That is,  $z = \alpha x + (1-\alpha)y  = \alpha x' + (1-\alpha)y'$ for some $x,x' \in A$ and $y,y' \in B$. For these elements, the cell of the partition in which they reside is not determined, but, it turns out not to matter. See the appendix for the formal argument.}

	\begin{ax}{A}{Translation Invariance}{ti}
		For $A,B \in \D$, we have 
		$$ (A,\{P_1, \ldots P_n\}) \sim (\alpha A + (1-\alpha)B,\{Q_1, \ldots Q_n\}),$$
		whenever $Q_i \subseteq (\alpha P_i + (1-\alpha)B )$ for all $i \leq n$.
	\end{ax}

	Recall that \r{ax:si} also implies that the value of identification should not depend on zero probability perturbations. The following axiom reflects this implication.
	
	\begin{ax}{A}{Belief Consistency}{rcr}
		Fix $A \in \D$,  and let  $\{P_1, P_2, \ldots P_n\}$ be a partition of $A$ such that $\mu(W_{A,P_1}) = 0$. Then:
		$$(A,\{P_1, P_2, \ldots P_n\}) \sim (A,\{P_1 \cup P_2, \ldots P_n\}).$$
	\end{ax}
	
	
	Within the expected utility framework, translation invariance and belief consistency are equivalent to structural invariance. Thus, along with the other axioms from the general model, the two axioms above provide a characterization of expected identification value maximization with linear utility. 
	
	\begin{theorem}
		\label{thm:eu-rep}
		Let $\s$ be defined over $\Pi(\E^\Delta)$. Then $\s$ satisfies \r{ax:ti} and \r{ax:rcr} if and only if it satisfies \r{ax:si}. 
	\end{theorem}

	\section{Shannon Entropy}
	\label{sec:entropy}
	Aside from applying to broader settings, out analysis can be used as a building block to provide foundations for specific theories of Bayesian experimental design. In this section we show how our Structural Invariance and Identification Separability can be strengthened to characterize the case in which the identification index $\tau$ conforms to the Shannon entropy:
	\begin{align*}
		\tau(W)=-\log(\mu(W)).
	\end{align*}
	Notice that within this special case, the value of identifying a subset of utilities depends only on it's ex-ante probability. 
	
	First, Structural Invariance can be strengthened to a Symmetry axiom stating that experiments inducing more evenly distributed probabilities across observations are preferred. When the analyst's value for identification depends only on its prior probability, then experiments in which the probability of each observation is approximately equal ensure that the ex-post identification value is approximately equal as well. Thus, for a cautious analyst, such experiments are desirable as they increase the worst case identification.
	
	\begin{ax}{A}{Symmetry}{sym}
		Fix $A,B \in \D$,  and partitions $\{P_1, \ldots P_n\}$ and $\{Q_1, \ldots Q_n\}$ of $A$ and $B$, respectively. Then if $|\mu(W_{B,Q_i}) - \tfrac{1}{n}| \geq |\mu(W_{A,P_i}) - \tfrac{1}{n}|$ for $i \leq n$, it follows that
		$$(A,\{P_1 \ldots P_n\}) \s (B,\{Q_1\ldots Q_n\}).$$
	\end{ax}
	
	Notice that if two experiments induce $\mu$-equivalent identification sets then they also induce the same distribution over the set of positive probability observations. Under \r{ax:rcr}, we can ignore $\mu$-probability zero observations, and so symmetry ensures that the experiments are equally valued. In other words, Symmetry (along with Belief Consistency) imply Structural Invariance.
	
	Next, we can strengthen Identification Separability to get not only additivity but the specific logarithmic form of the entropic representation.
	Entropic Additivity, below, disciplines how much the analyst values replacing $P_1$ in $(A,\P)$ with one of it partitions. Let $\P = \{P_1, \ldots P_n\}$ be a partition $A$ and $\P_1 = \{P^1_1, \ldots P^1_k\}$ a partition $P_1$. Then $\P^\dag = \{P^1_1, \ldots P^1_k, P_2, \ldots P_n\}$ is also partition of $A$. Fix an experiment $(B,\Q)$ such that 
	$\mu(W_{B,Q_i}) = \mu(W_{A,P^1_i} \mid W_{A,P_1})$ for $i=1,...,k$.
	
	The fundamental character of the entropic representation is that the value of an experiment only depends on the ratio between the prior and each induced posterior: as such learning which element of $\Q$ was chosen would impart the same value to the analyst as learning which element of $\P_1$ was chosen \emph{conditional on already knowing that $P_1$ was chosen from $\P$}. Further, the partition $\P^\dag$ is exactly like learning $\P$ and in the event $P_1 \in \P$ is chosen further learning which element of $\P_1$ is chosen. The extra learning happens with probability $\mu(W_{A,P_1})$: so the value $(A,\P^\dag)$ should equal the value of $(A,\P)$ plus $\mu(W_{A,P_1})$ times the value of learning the element chosen from $\P_1$, which as argued above is the value of $(B,\Q)$. Translating this into lotteries, we have:

	
	
		
	
	\begin{ax}{A}{entropic additivity}{ea}
		Fix $A\in \D$ let $\P = \{P_1, \ldots P_n\}$ partition $A$ and let $ \{P^1_1, \ldots P^1_k\}$ partition $P_1$. So $\P^\dag = \{P^1_1, \ldots P^1_k, P_2, \ldots P_n\}$ is also partition of $A$. Set $\alpha = \tfrac{1}{1 + \mu(W_{A,P_1})}$.
		Then if $B \in \D$ is such that $\Q = \{Q_1, \ldots Q_k\}$ is a partition of $B$ with $\mu(W_{B,Q_i}) = \mu(W_{A,P^1_i} \mid W_{A,P_1})$, it follows that
		
		$$
		\alpha (A, \P^\dag) + (1-\alpha)(A, \{A\}) \sim \alpha(A, \P) + (1-\alpha)(B, \Q)
		$$
	\end{ax}
	
	By replacing Structural Invariance and Identification Separability with the stronger entropic variants above, we find a characterization of expected entropy minimization.
	
	\begin{theorem}
		\label{thm:ent_rep}
		Let $\mu$ be non-atomic.
		The preference $\s$ satisfies \r{ax:eu}--\r{ax:mon} and \r{ax:rcr}--\r{ax:ea} if and only if it is represented by
		\begin{align*}
			F(\pi)=-\sum_{\supp(\pi)}\Big(\sum_{P \in \P} \log(\mu(W_{A,P})) \mu(W_{A,P})\Big)\pi(A,\P).
		\end{align*}
	\end{theorem}
	
	While the Shannon specification has significant normative appeal,
	Symmetry does impose restrictions on the analyst's risk attitudes that need to be spelled out. To illustrate, consider two experiments $(A,\{P_1,P_2\})$ and $(B, \{Q_1,Q_2\})$. Suppose
	\begin{align*}
		\mu(W_{A,P_1})&=\mu(W_{A,P_2})=\frac{1}{2}
		\\
		\mu(W_{B,Q_1})&=\frac{3}{4}, \mu(W_{B,Q_2})=\frac{1}{4}.
	\end{align*}
	Thus, if the analyst offers $(A,\{P_1,P_2\})$ she will be able to rule out \say{half} of the preference for the subject regardless of the subject's true utility. However, if she offers $(B,\{Q_1,Q_2\})$, then the size of the mass of preference she will be able to eliminate depends on the subjects preference. If the subject's preference is maximized in $Q_1$, she will be able to eliminate three quarters; if it is maximized in $Q_2$, she will only eliminate one quarter. 
	The entropic model imposes that the former is preferred, implicitly requiring a specific risk preference on the part of the analyst. We view this as beyond the scope of what can be argued only on normative grounds.
	
	\section{Belief Free Models}
	\label{sec:belief-free}
	
	In many cases, the analyst may not entertain a prior over $\U$. Nonetheless, our theory applies almost exactly. In this case, the value function $F$ can be written as 
	\begin{equation}
		\label{eq:belief_free}
		F(A,\P)=\sum_{P \in \P} \nu(W_{A,P}),
	\end{equation}
	for an abstract subadditive identification index $\nu$ which does not separate the intrinsic value of identification from its likelihood. This is akin to the failure of separation into tastes and beliefs in state-dependent expected utility. 
	
	In the original model, the value of identification was invariant to $\mu$-measure zero perturbations. This is what allowed us to work with $\mu$-equivalent-approximations of partitions of $\U$, greatly extending the set of scope of application. Without beliefs, we must re-cast the notation of null sets in a more general from. Call $V \in \W$ \emph{transparent} if for any $(A, \{P_1, P_2, \ldots P_n\})$ with $W_{A,P_1} = V$, we have 
	\begin{equation}
		\label{eq:transparent_text}
		(A, \{P_1, P_2, \ldots P_n\}) \sim (A, \{P_1 \cup P_2, \ldots P_n\}).
	\end{equation}
	
	Using transparency as a preference based definition of nullness, we can restate the condition \ref{exp:nooverlap} and richness without having to appeal to beliefs. In particular, assume that for a given set of experiments $\E$
	
	\begin{enumerate}[label=(\texttt{E\arabic*'}), itemsep=-.4em]
		\item\label{exp:nooverlapT} For all $(A,\P) \in \E$ and $P,Q \in \P$, if $V \subseteq (W_{A,P} \cap W_{A,Q})$, then $V$ is transparent.
	\end{enumerate}
	
	Further, call two (finite) collections of subsets of $\U$, $\{W_1, \ldots W_n\}$ and $\{V_1, \ldots W_m\}$ \emph{$T$-equivalent} if for every non-transparent $W_i$ there exists a $V_j$ such that $W_i \setminus V_j$ and $V_j \setminus W_i$ are both transparent, and vice versa. That is, the collections are $T$-equivalent if they can be identified up to transparent sets.
	
	Modulo these two changes, Theorem \ref{thm:rep} goes through exactly as stated to arrive at a representation of the form \eqref{eq:belief_free}. To see this, notice that the set of all $\{V \mid V \subseteq W_{A,P} \cap W_{A,Q}, \text{ for some } (A,\P) \in \E, P,Q \in \P \}$ is a down-set. The ideal generated by this down-set is a subset of all transparent sets (it is immediate from their definition that transparent sets are closed under finite unions). Thus, there exists a $\{0,1\}$-valued finitely additive measure on $\W$ sending all such sets to 0. We can then define $\mu$ as this measure and set $\nu = \mu \cdot \tau$.
	\section{Observability Constraints}
	\label{sec:partial-observability}
	We conclude the paper by illustrating how our framework is general enough to capture partial observability in dynamic environments. We begin by considering the case in which an analyst employs an adaptive method and the subject is an expected utility maximizer. 
	
	Suppose an analyst first offers $\{x_{0},y_{0}\}$. If the subject chooses $x_{0}$, then she offers $\{x_{x},y_{x}\}$, otherwise she offers $\{x_{y},y_{y}\}$. Given the nature of the procedure, if the agent chooses $x$ ($y$) from $\{x,y\}$, then the analyst will know the choice out of $\{x_{x},y_{x}\}$ but not the choice out of $\{x_{y},y_{y}\}$ ($\{x_{x},y_{x}\}$). Figure \ref{fig:adaptive} illustrates the procedure.
	
	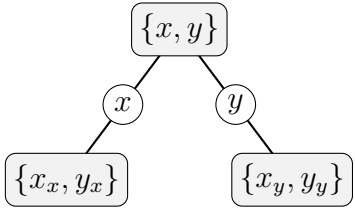
\begin{figure}
		\centering
		
		\begin{tikzpicture}
			
			\def\xshift{1.5cm} 
			\def\yshift{2cm} 
			
			\tikzstyle{node_shaded} = [draw, fill=gray!10, rectangle, rounded corners, minimum size=20pt, inner sep=3pt]
			\tikzstyle{circle_node} = [draw, circle, fill=white, minimum size=15pt, inner sep=1pt]
			
			\coordinate (center) at (0,0);
			
			\draw[thick] (center) -- ++(-\xshift, -\yshift); 
			\draw[thick] (center) -- ++(\xshift, -\yshift);  
			
			\node[node_shaded] at (center) {$\{x,y\}$};
			\node[node_shaded] at ($(center) + (-\xshift, -\yshift)$) {$\{x_{x},y_{x}\}$};
			\node[node_shaded] at ($(center) + (\xshift, -\yshift)$) {$\{x_{y},y_{y}\}$};
			
			\node[circle_node] at ($(center) + (-0.5*\xshift, -0.5*\yshift)$) {$x$};
			\node[circle_node] at ($(center) + (0.5*\xshift, -0.5*\yshift)$) {$y$};
			
		\end{tikzpicture}
		
		\caption{Adaptive Procedure}
		\label{fig:adaptive}
	\end{figure}
	
	Observe that  because of linearity of the preferences, observing a choice from $\{x,y\}$ and $\{x_{x},y_{x}\}$ is the same as observing a choice from 
	\begin{align*}
		A_{x}=\{\frac{1}{2}x+\frac{1}{2}x_{x},\frac{1}{2}x+\frac{1}{2}y_{x},\frac{1}{2}y+\frac{1}{2}x_{x},\frac{1}{2}y+\frac{1}{2}y_{x}\}.
	\end{align*}
	The reason is that any expected utility maximize that would choose $a$ out of $\{x,y\}$ and $b$ out of $\{x_{x},y_{x}\}$ would choose $\frac{1}{2}a+\frac{1}{2}b$ out of $A_{x}$.
	
	Similarly, observing a choice from $\{x,y\}$ and $\{x_{y},y_{y}\}$ is the same as observing a choice from 
	\begin{align*}
		A_{y}=\{\frac{1}{2}x+\frac{1}{2}x_{x},\frac{1}{2}x+\frac{1}{2}y_{y},\frac{1}{2}y+\frac{1}{2}x_{y},\frac{1}{2}y+\frac{1}{2}y_{y}\}.
	\end{align*}
	Consider the menu $A_{x}\cup A_{y}$ and the partition
	\begin{align*}
		\P=\Big\{\{\frac{1}{2}x+\frac{1}{2}x_{x},\frac{1}{2}x+\frac{1}{2}y_{x}\},\{\frac{1}{2}y+\frac{1}{2}x_{y},\frac{1}{2}y+\frac{1}{2}y_{y}\}\Big\}.
	\end{align*}
	Then the information provided by the adaptive design is equivalent to the information provided by $(A_{x}\cup A_{y},\P)$. 
	
	The above example can be easily generalized to adaptive procedures that employ $T$ rounds of choices that allow for non-binary menus. While the intuition is clear, precisely stating an equivalence result requires a significant amount of notation and so we leave this at the informal level.
	
	Next, consider an analyst interested in learning the subject's preference by employing a dynamic game. Suppose the game features two players, the subject and a computer. 
	The subject can first choose \emph{out} $(o)$ or \emph{in} $(i)$. If the subject chooses in, then the computer randomizing between \emph{right} $(r)$ and \emph{left} $(l)$. Following left the subject has a choice between $a$ and $b$ and following right, a choice between $c$ and $d$.
	The game is provided in Figure \ref{fig:game}.
	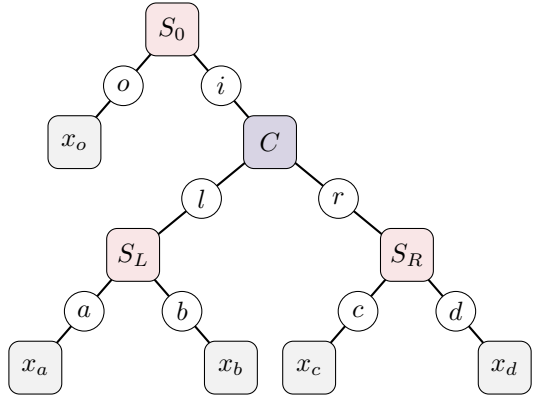
\begin{figure}
		\centering
		\begin{tikzpicture}[every node/.style={font=\footnotesize}]]
			\def\xshift{1.3cm} 
			\def\yshift{1.5cm} 
			
			\tikzstyle{node_shaded} = [draw, fill=gray!10, rectangle, rounded corners, minimum size=20pt, inner sep=3pt]
			\tikzstyle{circle_node} = [draw, circle, fill=white, minimum size=15pt, inner sep=1pt]
			
			\coordinate (center) at (0,0);
			\coordinate (c) at ($(center) + (\xshift, -\yshift)$);
			\coordinate (l) at ($(c) + (-1.4*\xshift, -\yshift)$);
			\coordinate (r) at ($(c) + (1.4*\xshift, -\yshift)$);
			
			\draw[thick] (center) -- ++(-\xshift, -\yshift); 
			\draw[thick] (center) -- ++(\xshift, -\yshift);  
			\draw[thick] (c) -- (l);  
			\draw[thick] (c) -- (r);  
			
			\draw[thick] (l) -- ++( - \xshift, -\yshift);  
			\draw[thick] (l) -- ++(\xshift, -\yshift);  
			\draw[thick] (r) -- ++( - \xshift, -\yshift);  
			\draw[thick] (r) -- ++(\xshift, -\yshift);  
			
			\node[node_shaded, fill=lam2!10] at (center) {$S_{0}$};
			\node[node_shaded] at ($(center) + (-\xshift, -\yshift)$) {$x_o$};
			\node[node_shaded, fill=lam3!20] at (c) {$C$};
			
			\node[node_shaded,fill=lam2!10] at (l) {$S_L$};
			\node[node_shaded,fill=lam2!10] at (r) {$S_R$};
			
			\node[node_shaded] at ($(l) + (-\xshift, -\yshift)$) {$x_{a}$};
			\node[node_shaded] at ($(l) + (\xshift, -\yshift)$) {$x_{b}$};
			\node[node_shaded] at ($(r) + (-\xshift, -\yshift)$) {$x_{c}$};
			\node[node_shaded] at ($(r) + (\xshift, -\yshift)$) {$x_{d}$};

			\node[circle_node] at ($(center) + (-0.5*\xshift, -0.5*\yshift)$) {$o$};
			\node[circle_node] at ($(center) + (0.5*\xshift, -0.5*\yshift)$) {$i$};
			
			\node[circle_node] at ($(c) + (-0.5*1.4*\xshift, -0.5*\yshift)$) {$l$};
			\node[circle_node] at ($(c) + (0.5*1.4*\xshift, -0.5*\yshift)$) {$r$};
			
			\node[circle_node] at ($(l) + (-0.5*\xshift, -0.5*\yshift)$) {$a$};
			\node[circle_node] at ($(l) + (0.5*\xshift, -0.5*\yshift)$) {$b$};
			\node[circle_node] at ($(r) + (-0.5*\xshift, -0.5*\yshift)$) {$c$};
			\node[circle_node] at ($(r) + (0.5*\xshift, -0.5*\yshift)$) {$d$};
			
		\end{tikzpicture}
		\caption{The dynamic game between the subject and computer. The subject's decision nodes are shaded \textcolor{lam2!50}{red} and the computer's \textcolor{lam3!70}{blue}.}
		\label{fig:game}
	\end{figure}
	
	Suppose that the analyst only observes the on-path strategies; she cannot know how the subject would behave in a sub-game that is not reached. There are five actions the analyst could potentially observe: $(o)$, $(i,a)$, $(i,b)$, $(i,c)$, and $(i,d)$. Notice which of these is observed depends not only on the subject's choice but also the outcome of the computer randomization. These observations can be adapted into our framework. Let $A$ be the set of all strategies for the subject:
	\begin{align*}
		A=
		\left\{\begin{array}{l}
			(i,a,c),(i,b,c),(i,a,d),(i,b,d), \\
			(o,a,c),(o,a,d),(o,b,c),(o,b,d)
		\end{array}\right\}
	\end{align*}
	Now consider the following partitions of $A$
	\begin{align*}
		\P_L=\left\{\begin{array}{c}
			\left\{\begin{array}{c}
				(i,a,c),(i,b,c), \\
				(i,a,d),(i,b,d) 
			\end{array}\right\},\\[3ex]
			\left\{\begin{array}{c}
				(o,a,c),(o,a,d)
			\end{array}\right\},\\[2ex]
			\left\{\begin{array}{c}
				(o,b,c),(o,b,d)
			\end{array}\right\}\phantom{,}
		\end{array}\right\}
		&&
		\P_R=\left\{\begin{array}{c}
			\left\{\begin{array}{c}
				(i,a,c),(i,b,c), \\
				(i,a,d),(i,b,d) 
			\end{array}\right\},\\[3ex]
			\left\{\begin{array}{c}
				(o,a,c),(o,b,c)
			\end{array}\right\},\\[2ex]
			\left\{\begin{array}{c}
				(o,a,d),(o,b,d)
			\end{array}\right\}\phantom{,}
		\end{array}\right\}
	\end{align*}
	
	Observe that a single choice of $(A,\P_L)$ would yield the same information as observing the subject play the dynamic game in the event that the computer chooses \emph{left}, and likewise $(A,\P_R)$ should the computer choose \emph{right}. Thus, if the computer chooses \emph{left} with probability $\alpha$, then an observation in the dynamic game is observationally equivalent to an observation from the randomized experiment $\alpha(A,\P_L) + (1-\alpha)(A,\P_R)$. Again, this can be generalized: observations of on-path behavior in dynamic games can be incorporated into our framework via the appropriately constructed random experiments. 
	
	\appendix 
	\section{Proofs}
	\subsection{Proof of Theorem \ref{thm:rep}}
	
	Let $\part(\U)$ denote the finite $\W$-measurable partitions of $\U$. That is, $\{W_1, \ldots W_n\} \in \part(\U)$ if it is a (finite) partition of $\U$ such that each $W_i \in \W$.
	
	\begin{lemma}
		\label{lem:allmueqiv}
		Let $(A, \{P_1, \ldots P_n\}) \in \E$; then $\{W_{A, P_i}\}_{i \leq n}$ is $\mu$-equivalent to some partition $\WW \in \part(\U)$. 
	\end{lemma}
	
	\begin{subproof}
		By \ref{exp:measure} $\{W_{A, P_i}\}_{i \leq n}$ and since each $u \in \U$ find some maximum on $A$, we have $\bigcup_{i\leq n} W_{A, P_i} = \U$. Define $W_i = W_{A,P_i} \setminus \bigcup_{i < j} W_{A, P_j}$. Clearly, $\{W_i\}_{i \leq n} \in \part(\U)$. Moreover, 
		$\mu(W_{A,P_i}) = \mu(W_{A,P_i}) - \sum_{i < j} \mu(W_{A, P_i} \cap W_{A, P_j}) \leq \mu(W_{A,P_i} \setminus \bigcup_{i < j} W_{A, P_j}) =   \mu(W_i) = \mu(W_i \cap W_{A,P_i}) \leq \mu(W_{A,P_i})$, establishing $\mu$-equivalence (the first, and only non-set-theoretically obvious, equality comes form \ref{exp:nooverlap}).
	\end{subproof}
	
	\begin{lemma}
		\label{lem:bijection}
		Let $\bm{W} \subseteq \W$ be $\mu$-equivalent to $\bm{V} \subseteq \W$, and assume $\mu(W \cap W') = 0$ and $\mu(V_i \cap V_j) = 0$ for any distinct $W, W' \in \bm{W}$ and $V, V' \in \bm{V}$. Then there exists a bijection, $q$, between $\{W \in \bm{W} \mid \mu(W) > 0\}$ and $\{V \in \bm{V} \mid \mu(V) > 0\}$ such that $\mu(W) = \mu(W \cap h(W)) = \mu(q(W))$.
	\end{lemma}
	
	\begin{subproof}
		Take some $W \in \bm{W}$ with $\mu(W) > 0$. By $\mu$-equivalence, there exists a $V \in \bm{V}$ such that $\mu(W) = \mu(W \cap V) = \mu(V)$.
		To see that it is unique, let $V,V' \in \bm{V}$ both be such that the needed relation holds. 
		Then we have $\mu(W) < 2\mu(W) = \mu(V \cap W) + \mu(V' \cap W) = \mu((V \cup V') \cap W) + \mu((V \cap V') \cap W) \leq \mu(W) + \mu(V \cap V')$. This means that $\mu(V \cap V') > 0$. By the condition in the statement of the Lemma, this requires $V = V'$.
	\end{subproof}

	\begin{lemma}
		\label{lem:mucoarse}
		Let $\WW \in \part(\U)$ and let $(A, \P) \in \E$ be such that such that $\{W_{A, P} | P\in \P\}$ is $\mu$-equivalent it. Then there exists a family of partitions 
		$$\{ \P_{\VV} \mid \VV \text{ is a coarsening of } \WW\}$$
		such that 
		\begin{enumerate}
			\item $\{W_{A, P} | P\in \P_{\VV}\}$ is $\mu$-equivalent $\VV$, and 
			\item If $\VV'$ is a coarsening of $\VV$, then $\P_{\VV'}$ is a coarsening of $\P_{\VV}$.
		\end{enumerate}
	\end{lemma}
	
	\begin{subproof}
		First, for each $W \in \WW$ with $\mu(W) > 0$, let $P_W \in \P$ be the unique element such that $\mu(W_{A, P}) = \mu(W_{A, P} \cap W) = \mu(W)$. This exists by Lemma \ref{lem:bijection}.
		
		Now, for each $V \in \VV$, let $[V] = \{P_W \in \P \mid W \in \WW, \mu(W) \geq 0, W \subseteq V\}$. It is easy to see that $\{\bigcup_{P_W \in [V]} W_{A,P_W} \mid V \in \VV\}$ is $\mu$-equivalent $\VV$. Indeed, either $\mu(V) = 0$, in which case $[V] = \emptyset$ or $\mu(V) = \sum_{W \subseteq V} \mu(W) = \sum_{P_W \in [V]} \mu(W_{A,P_W}) =  \mu(\bigcup_{P_W \in [V]} W_{A,P_W})$.
		Let $P_{\VV}$ be the coarsest partition containing $\{\bigcup [V] \mid V \in \VV\}$. Note that $\bigcup [V] \cap \bigcup[V'] = \emptyset$ whenever $V \neq V'$, so $P_{\VV}$ will simply be $\{\bigcup [V] \mid V \in \VV\}$ adjoined with whatever elements of $A$ were not in any $\bigcup[V]$---these are exactly the observations that have 0-probability.
	\end{subproof}
	
	\begin{tproof}{thm:rep}
		From \r{ax:eu} there exists a vNM index $\vnm: \E \to \R$ such that
		$$
		\pi \s \rho \qquad \iff \qquad \sum_{\supp(\pi)}  \vnm(e) \pi(e) \geq \sum_{\supp(\rho)}  \vnm(e) \rho(e)
		$$
		$\vnm$ can be chosen such that $\vnm(\{x\}, \{\{x\}\}) = 0$ (since this induces the trivial partition of $\U$ independent of $x$, by \r{ax:si}, the choice of $x$ is irrelevant).
		
		For each $\WW \in \part(\U)$ let $(A_\WW, \P_\WW) \in \E$ be such that such that $\{W_{A_\WW, P} | P\in \P_\WW\}$ is $\mu$-equivalent to $\mathcal{W}$. This exists by the richness assumption on the set of experiments. Define the function $\phi: \part(\U) \to \R$ as 
		\begin{equation}
			\label{eq:phi}
			\phi: \WW \mapsto \vnm(A_\WW, \P_\WW).
		\end{equation}
		By Axiom \ref{ax:si}, $\phi$ does not depend on the choice of $(A_\WW, \P_\WW)$.
		
		Call $\nu: \W \to \R$ a \emph{GL-representation} of $\phi$ if $\nu(\emptyset) = 0$ and
		\begin{equation}
			\label{eq:rep0}
			\phi(\WW) = \sum_{W \in \WW} \nu(W).
		\end{equation}

		\begin{lemma}
			\label{lem:GL}
			A GL representation of $\phi$ exists.
		\end{lemma}
		
		\begin{subproof}
			Following \cite{gilboa1991value} (Observation 2.1), call two partitions, $\WW, \VV \in \part(\U)$, \emph{non-intersecting} iff there is an event $U \in \W$ such that  $U$ is measurable with respect to both $\WW$ and $\VV$ and such that $\WW|_U$ refines $\VV|_U$ and $\VV|_{U^c}$ refines $\mathcal{U}|_{U^c}$. 
			
			Theorem 3.2 of \cite{gilboa1991value} states that a GL-representation of $\phi$ exists if and only if 
			\begin{equation}
				\label{eq:partialcom}
				\phi(\WW \land \VV) + \phi(\WW \lor \VV) = \phi(\WW) + \phi(\VV)
			\end{equation}
			for any non-intersecting partitions, where $\WW \land \VV$ and $\WW \lor \VV$ denote their meet (coarsest common refinement) and join (finest common coarsening), respectively.
			
			So, let $\WW$ and $\VV$ be non-intersecting and $U$ the jointly measurable event delineating which partition is finer.
			Consider an experiment  $(A,\P_{\WW \land \VV})$ such that 
			$\{W_{A, P} | P\in \P_{\WW\land \VV}\}$ is $\mu$-equivalent to $\mathcal{W}$. Again, this exists by the richness assumption.
			
			By Lemma \ref{lem:mucoarse},we can construct partitions $\P_{\WW}$, $\P_{\VV}$, $\P_{\WW \lor \VV}$, $\P_{\{U,U^c\}}$ of $A$, each inducing a partition $\mu$-equivalent with respect to the corresponding partition of $\U$ (i.e., indicated by the subscript).
			
			Let $B \in \P_{\{U,U^c\}}$ be the cell such that $\mu(W_{A,B}) = \mu(W_{A,B} \cap U) = \mu(U)$. It follows that $B$ is measurable with respect to all of the above partitions, and furthermore, $(\P_{\WW \land \VV})_{B}(\P_{\WW \lor \VV}) = \P_{\WW}$ and $(\P_{\WW \lor \VV})_{B}(\P_{\WW \land \VV}) = \P_{\VV}$.
			
			Applying axiom \r{ax:is}, we have
			$$
			\tfrac12 (A, \P_{\WW \land \VV}) 
			+ 
			\tfrac12 (A, \P_{\WW \lor \VV}) 
			\sim 
			\tfrac12 (A, \P_{\WW}) 
			+
			\tfrac12 (A, \P_{\VV}) 
			$$
			Thus, from \eqref{eq:phi}, the definition of $\phi$, we obtain \eqref{eq:partialcom}. Theorem 3.2 of \cite{gilboa1991value} ensures us of the existence of some GL-representation.
		\end{subproof}

		Call $V \in \W$ \emph{transparent} if for any partition containing $V$, $\{V, W, U_1 \ldots U_n\}$, it follows that
		\begin{equation}
			\label{eq:transparent}
			\phi(\{V, W, U_1 \ldots U_n\}) = \phi(\{V \cup W, U_1 \ldots U_n\}).
		\end{equation}
		
		Notice, while it is easier to write this as a condition on $\phi$, it is completely determined by the preference.
		Let $\Omega^{\deg}$ collect the transparent measurable sets. 
		Notice that if $\nu'$ is any GL-representation of $\phi$ and $V \in \Omega^{\deg}$ and $W \in \W$ with $W \cap V = \emptyset$, it holds that 
		\begin{equation}
			\label{eq:transparent2}
			\nu'(W \cup V) = \nu'(W) + \nu'(V).
		\end{equation}
		This follows immediately from plugging the GL-representation, \eqref{eq:rep0}, into \eqref{eq:transparent}. In particular, notice that $\nu'$ is finitely additive over $\Omega^{\deg}$.
		
		\begin{lemma}
			\label{lem:addondeg}
			Let $\C \subseteq \deg$ be a set of transparent subsets such that (i) $\emptyset \in \C$, (ii) $\U \notin \C$, (iii) $W,V \in \C$ implies $W \cup V \in \C$, and, (iv) $W \in \C$ and $V \in \W$ with $V \subseteq W$ implies $V \in \C$. Then there exists a GL-representation, $\nu$, of $\phi$ such that $\nu(V) = 0$ for all $V \in \C$.
		\end{lemma}
		
		\begin{subproof}
			Let $\nu' : \W \to \R$ be a GL-representation, which exists by Lemma \ref{lem:GL}. Notice that $\C$ is a ring of sets, since if $W,V \in \C$ then $W \setminus V \subseteq W \in \C$ by (iv). Further, by \eqref{eq:transparent2}, it follows that $\nu'|_\C: \C \to \R$ is finitely additive. Hence, by Theorem 3.2.5 of \cite{rao1983theory}, there exists a finitely additive measure $\mu': \W \to \R$ that extends $\nu'|_\C$. 
			
			Notice also that $\C$ is an ideal in $\W$ (as a Boolean algebra of sets). Thus by the Boolean prime ideal theorem, $\C$ is contained in some maximal (proper) ideal, $\mathcal{I} \subset \W$. Then 
			$$
			\mu'': W \mapsto 
			\begin{dcases}
				0 &\text{ if } W \in \mathcal{I} \\
				\mu'(\U)  &\text{ otherwise}
			\end{dcases}
			$$
			is a finitely additive measure. It follows that $\mu^\dag = \mu' - \mu''$ is a finitely additive measure and $\mu^\dag(\U) = 0$. By Proposition 3.3 of \cite{gilboa1991value},  $\nu = \nu' - \mu^\dag$ is also a GL-representation of $\phi$. Moreover, for $V \in \C$, we have $\nu(V) = \nu'(V) - \mu^\dag(V) = \nu'(V) - \mu'(V) + \mu''(V) =  \nu'(V) - \nu'(V) + 0 = 0$.
		\end{subproof}
		
		Let $\C^{null} = \{W \in \W \mid \mu(W) = 0\}$. Note that $\C^{null}$ satisfies the conditions for Lemma \ref{lem:addondeg} (that $\C^{null} \subseteq \Omega^{\deg}$ is a straightforward consequence of \r{ax:si}). Thus, we can choose $\nu$ such that $\mu(W) = 0$ implies $\nu(W) = 0$. Now define $\tau$ as
		$$
		\tau: W \mapsto 
		\begin{cases}
			\tfrac{\nu(W)}{\mu(W)} &\text{ if } W \notin \C^{null} \\
			0 &\text{ otherwise.} 
		\end{cases}
		$$
		It is obvious that this $\tau$ satisfies \ref{tau:empty}. To see that it also $\tau$ satisfies \ref{tau:subadd}, first note that by \r{ax:mon}, if $\WW$ refines $\VV$ then $\phi(\WW) \geq \phi(\VV)$, and so by Observation 4.1 of \cite{gilboa1991value} $\nu$ is subadditive.
		
		Now let $\mu(V) > 0$ and $W \subset V$; by sub-additivity $\nu(W) + \nu(V \setminus W) \geq \nu(V)$. Plugging in for the definition of $\tau$ we have $\tau(W)\mu(W) + \tau(V \setminus W)\mu(V \setminus W) \geq \tau(V)\mu(V)$. Dividing by $\mu(V)$ delivers the inequality part of \ref{tau:subadd}. If we further assume that $\mu(W) = 0$, then $W \in \W^{\deg}$ and the equality part of \ref{tau:subadd} follows from the definition of transparency \eqref{eq:transparent2}.
		
		Finally, to see that $\tau$ represents $\s$ according to \eqref{eq:rep2}, let $(A,\Q) \in \E$. By Lemma \ref{lem:allmueqiv}, there is some $\WW \in \part(\U)$ such that $\{W_{A, Q}\}_{Q \in \Q}$ is $\mu$-equivalent to $\WW$. Let $(A_{\WW}, \P_{\WW})$ be the experiment used to define $\phi(\W)$. Clearly, $(A,\Q)$ and $(A_{\WW}, \P_{\WW})$ are themselves $\mu$-equivalent, and hence by \r{ax:si}, $(A,\Q) \sim (A_{\WW}, \P_{\WW})$.
		
		By $\mu$-equivalence, for each $Q \in \Q$ with $\mu(W_{A,Q}) > 0$ there exists some $W^Q \in \WW$ such that $\mu(W_{A,Q}) = \mu(W^Q \cap W_{A,Q}) =  \mu(W^Q)$. In particular, this implies that $\mu(W^Q \setminus W_{A,Q}) = \mu(W_{A,Q} \setminus W^Q) = 0$. This further implies, via the equality part of \ref{tau:subadd}, that $\tau(W^Q) = \tau(W^Q \cup W_{A,Q}) = \tau(W_{A,Q})$. So finally, we have
		\begin{align*}
			\sum_{Q \in \Q} \tau(W_{A,Q}) \mu(W_{A,Q})
			&= \sum_{\substack{Q \in \Q \\ \mu(W_{A,Q}) > 0}} \tau(W^Q) \mu(W^Q)\\
			&= \sum_{W \in \WW} \nu(W)\\
			&= \phi(\WW)\\
			&= \vnm(A_{\WW},P_{\WW})\\
			&= \vnm(A,Q)
		\end{align*}
		as desired.
	\end{tproof}

	\subsection{Proof of Theorem \ref{thm:rich}}
	\subsubsection*{Preliminaries: Convex Spaces}
	
	For a set $X \subseteq \R^\ell$, let $\conv(X)$, $\int(X)$, and $\cl(X)$ denote the convex hull, the interior, and the closure of $X$, respectively. If $X$ is convex, then  $\ext(X)$ collects all the extreme points of $X$ and $\ri(X)$ denotes the relative interior of $X$. When it is not confusing to do so, we will write $\ri(X)$ and $\ext(X)$ to mean $\ri(\conv(X))$ and $\ext(\conv(X))$ for non-convex $X$. 
	
	For convex $X$, let $F \subset X$ be called a face if whenever $\alpha x + (1-\alpha)y \in F$ (for $x,y \in X$) then also $x,y \in F$. Let $\face(X)$ denote the set of all (non-empty) faces of $X$ and $\face^\circ(X) = \{\ri(F) \mid F \in \face(X)\}$.
	
	If $X \subseteq \R^\ell$ is a convex set and $\ext(X)$ is finite then $X$ is a called at polytope. Let $\poly$ denote the set of all polytopes in $\R^\ell$. If $X \in \poly$, then $\face(X)$ is finite.
	
	
	If $K \subseteq \R^\ell$ and $\lambda K \subseteq K$ for all $\lambda \geq 0$ then $K$ is called a \emph{cone}. We say a cone $K$ is generated by $X$ if $K = \{\lambda x \mid x \in X, \lambda \geq 0\}$. A cone $K$ is \emph{polyhedral} if it is generated by a polytope; let $\K^*$ denote all such cones. Let $\K$ denote the set of pointed polyhedral cones, those cones with $\0 \in \ext(K)$. The face of a polyhedral cone is a polyhedral cone.
	
	For $X \in \poly$, let 
	$$X^\star = \bigcup_{I \subseteq \ext(X)} \sum_{i\in I} \tfrac{x_i}{|I|}$$
	The set $X^\star$ is a decision problem that contains one point in the relative interior of every face of $X$. Further, given a partition of $\H = \{H_1 \ldots H_n\}$ of $\face(X)$, let $\H^\star = \{H^\star_1 \ldots H^\star_n\}$ denote the partition of $X^\star$ defined via $H^\star = X^\star \cap (\bigcup_{F \in H_i} \ri(F))$.
	
	For $X \subseteq \R^\ell$ (convex or not) and $x \in X$ let $N(X,x) = \{u \in \U \mid u(y-x) \leq 0, \text{ for all } y \in X\}$ denote the normal cone of $X$ at $x$. Alternatively, $N(X,x) = \{u \in \U^\Delta \mid x \in \argmax_{X} u\}$. Notice that $W_{A,B} = \bigcup_{x \in B} N(A,x)$. 
	
	
	For $X\in \poly$, and a face $F \in \face(X)$, let $N(X,F) = \bigcap_{x \in F} N(X, x)$. It follows that $N(X,F) = \{u \in \U^\Delta \mid F \subseteq \argmax_{x \in D} u(x)\} = N(X,x)$ for any $x \in \ri(F)$. Notice therefore that
	\begin{equation}
		\label{eq:isopartitions}
		\bigcup_{x \in H^\star_i} \ri(N(X^\star,x)) = \bigcup_{F \in H_i} \ri(N(X,F)) 
	\end{equation}
	
	It is immediate that $N(X,F)$ is closed and in $\K^*$. Let $\NN(X) = \{N(X,F) \mid F \in \face(X)\}$ denote the normal fan of $X$; that $\NN(X)$ is a fan indicates that it is a family of cones such, with the following two properties:
	\begin{enumerate}[label=(\roman*), itemsep=-.5em]
		\item Every nonempty face of a cone in $\NN(X)$ is also a cone in $\NN(X)$.
		\item The intersection of any two cones in $\NN(X)$ is a face of both. 
	\end{enumerate}
	Furthermore, $\NN(X)$ is complete (the union of $\NN(X)$ is $\U^\Delta$).
	Let $\NN^\circ(X) = \{\ri(N(X,F)) \mid F \in \face(X)\}$. 
	
	\begin{lemma}
		\label{lem:convexfacts}
		The following are true for all convex $X$ and $Y$:
		
		\begin{enumerate}[itemsep=-.2em]
			\item \label{relintclosure}  $\cl(\ri(X)) = \cl(X)$ (Theorem 6.3 of \cite{rockafellar1970convex}).
			\item \label{facepartition} $\face^\circ(X)$ is a partition of $X$ (Theorem 18.2 of \cite{rockafellar1970convex}). 
			\item \label{facestouch} Let $F \in \face(X)$ and $Y \subseteq X$ be such that $\ri(Y)\cap F \neq \emptyset$, then $Y\subseteq F$. (Theorem 18.1 of \cite{rockafellar1970convex}).  
		\end{enumerate}
	\end{lemma}

	\begin{lemma}
		\label{lem:fanpartitions}
		For $X \in \poly$, $\NN^\circ(X)$ is a partition of $\U^\Delta$. 
	\end{lemma}
	
	\begin{proof}
		Let $u \in \U^\Delta$. Since $\NN(X)$ is complete, $u \in K$ for some $K \in \NN(X)$. By property (i) of fans, we see that $\face(K) \subseteq \NN(X)$; since $\face^\circ(K)$ is a partition of $K$, it follows that $x \in \ri(F)$ for some $F \in \NN(X)$. So, the elements of $\NN^\circ(X)$ cover $\U^\Delta$.
		
		Now assume that $x \in \ri(K) \cap \ri(K')$ for some $K,K' \in \NN(X)$. Then by property (ii) of fans, $K \cap K' \in \face(K)$; moreover, since $x \in (K\cap K') \cap \ri(K) \neq \emptyset$, Lemma \ref{lem:convexfacts}.\ref{facestouch} delivers that the face $K \cap K'$ must be equal to $K$ itself. By symmetry also $K' = (K \cap K') = K$. So, the elements of $\NN^\circ(X)$ are disjoint.
	\end{proof}
	
	\begin{lemma}
		\label{lem:fanpartitions}
		For polytopes $X$ and $X'$, the following are equivalent 
		\begin{enumerate}[label=(\roman*), itemsep=-.1em]
			\item $X = \alpha X' + Z$, for some polytope $Z$ and $\alpha > 0$ 
			\item for all $K \in \NN(X)$ there is a $K' \in \NN(X')$ such that $K \subseteq K'$
			\item $\NN^\circ(X)$ refines $\NN^\circ(X')$.
		\end{enumerate}
	\end{lemma}
	
	\begin{proof}
		(i) $\leftrightarrow$ (ii) Theorem 15.1.2 of \cite{grnbaum2003convex} 
		
		(ii) $\rightarrow$ (iii) 
		Take some $\ri(K) \in \NN^\circ(X)$. Let $K^\dag= \bigcap \{K' \in \NN(X') \mid K \subseteq K'\}$, which is an element of $\NN(X')$ by (ii) and the properties of fans. Moreover, by construction $K \centernot{\subseteq} F$ for any $F \in \face(K^\dag)$ with $F \subsetneq K^\dag$. Thus, by (the contra-positive of) Lemma \ref{lem:convexfacts}.\ref{facestouch}, we have that $\ri(K) \cap \ri(F) \neq \emptyset$ for all such $F$.  Then, since $\face^\circ(K^\dag)$ partitions $K^\dag$, it follows that $\ri(K) \subseteq \ri(K^\dag)$.
		
		(iii) $\rightarrow$ (ii) Take some $K \in \NN(X)$. Then by (iii) $\ri(K) \subseteq \ri(K')$ for some $K' \in \NN(X')$.
		Since $K$ and $K'$ are closed, we have $K = \cl(K)$ and $K' = \cl(K')$. Thus, $K = \cl(\ri(K)) \subseteq \cl(\ri(K'))  = K'$, where both equalities come from Lemma \ref{lem:convexfacts}.\ref{relintclosure} and the inclusion relation from the fact that taking closures is subset preserving. 
	\end{proof} 
	
	
	
	\begin{tproof}{thm:rich}
		The coarsening property is obvious. We will show that any partition can be captured up to $\mu$-equivalence. 
		Let $\WW = \{W_1, \ldots W_n\} \in \part(\U)$.  First, from \cite{gul2006random} Proposition 6(ii), we can write each $W_i \in \W$  as the finite union of elements in $\K$: $W_i = \bigcup_{j = 1}^{m_i} \ri(K_i^j)$. Moreover, by \cite{gul2006random} Proposition 4, each $K_i^j = N(X_i^j,x_i^j)$ for some polytope $X_i^j$ and $x_i^j \in X_i^j$. Thus $\ri(K_i^j) \in \NN^\circ(X_i^j)$.
		
		Let $a = m_1 + \ldots + m_n$ and consider the polytope $X = \sum_{i = 1}^{n} \sum_{j = 1}^{m_i} \tfrac{1}{a} X_i^j$. By Lemma \ref{lem:fanpartitions}, $\NN^\circ(X)$ refines each $\NN^\circ(X_i^j)$. Let $\H = \{H_1, \ldots, H_n\}$ be a partition of $\face(X)$ defined by
		\begin{equation}
			\label{eq:partiondef}
			H_i = \{ F \in \face(X) \mid \ri(N(X,F)) \subseteq W_i\}
		\end{equation}
		
		Now take some $i \leq n$ and $u \in W_i$. So, there exists some $j \leq m_i$ such that $u \in \ri(K_i^j) \in  \NN^\circ(X_i^j)$. Since $\NN^\circ(X)$ is a partition of $\U$, there exists some $F \in \face(X)$ with $u \in \ri(N(X,F))$, and furthermore, since this partition refines $\NN^\circ(X_i^j)$, $\ri(N(X,F)) \subseteq \ri(K_i^j) \subseteq W_i$. Hence $F \in H_i$ and so $u \in \bigcup_{F \in H_i} \ri(N(X,F))$. We have established that $W_i \subseteq \bigcup_{F \in H_i} \ri(N(X,F))$, and since the other inclusion is obvious, that $W_i = \bigcup_{F \in H_i} \ri(N(X,F))$.
		Now, on the basis of \eqref{eq:isopartitions}, we have
		\begin{equation}
			\label{eq:facepartition}
			W_i = \bigcup_{x \in H^\star_i} \ri(N(X^\star,x))
		\end{equation}
		Finally, by Lemma 2 of \cite{gul2006random}, we know that for $\mu$ which satisfies \ref{exp:nooverlap}, it must be that $\mu(\ri(N(X^\star_\WW,H^\star)) = \mu(N(X^\star_\WW,H^\star).$ Thus, $\{W_{X^\star, H^\star}\}_{H^\star \in \H^\star}$ is $\mu$-equivalent to $\WW$.
	\end{tproof}
	
	\subsection{Proof of Theorem \ref{thm:eu-rep} and \ref{thm:ent_rep}}	
	\begin{tproof}{thm:eu-rep}
		Let $(A,\{P_1, \ldots P_n\})$ and $(B,\{Q_1, \ldots Q_m\})$ be such that $\{W_{A,P_i}\}_{i \leq n}$ is $\mu$-equivalent to $\{W_{B,Q_i}\}_{i \leq m}$.
		Furthermore, from Lemma \ref{lem:bijection}, we can assume there are $1 \leq k \leq n$ elements of each partition with positive $\mu$-probability and for each $i \leq k$, $\mu(W_{A,P_i}) = \mu(W_{A,P_i} \cap W_{B,Q_i}) = \mu(W_{B,Q_i})$.
		
		Consider the problem $C = \frac12A + \frac12B$. For each $i \leq n$, define $R_i \subseteq C$ as $R_i = \{\frac12P_i + \frac12B\} \cap \ext(C)$. 
		Clearly, we have for each $i \leq n$, $R_i \subseteq \frac12P_i + \frac12B$; it follows from \r{ax:ti} that $(A,\{P_1, \ldots P_n\}) \sim (C,\{R_1, \ldots R_n\})$. 
		
		Now for each $i \leq k$, let $R'_i = R_i \cap (\frac12P_i + \frac12Q_i) \cap \ext(C) = (\frac12P_i + \frac12Q_i) \cap \ext(C)$. The final equality arises from the fact that each extreme point of $C$ has a unique decomposition as elements of $A$ and $B$ (so that any $x \in (\frac12P_i + \frac12Q_i) \cap \ext(C)$ was not in $R_j$ for $j < i$).
		We claim that $\mu(W_{C, R_i \setminus R'_i}) = 0$. Indeed, 
		\begin{align*}
			W_{C, R_i \setminus R'_i} &\subseteq W_{A,P_i} \cap \bigcup_{j\neq i} W_{B,Q_j} \\
			&= \bigcup_{j\neq i} (W_{A,P_i} \cap W_{B,Q_j})  \\
			&\subseteq \bigcup_{j\neq i}  \big( (W_{A,Q_i} \cap W_{B,Q_j}) \cup (W_{A,P_i} \setminus W_{B,Q_i}) \big)
		\end{align*}
		The claim then follows from the fact that for all $i \neq j$, $\mu(W_{A,Q_i} \cap W_{B,Q_j}) = 0$ (from \ref{exp:nooverlap}) and $\mu(W_{A,P_i} \setminus W_{B,Q_i}) = 0$ (from $\mu$-equivalence).
		
		By repeatedly appealing to \r{ax:rcr}, we can see that 
		$$(C,\{R_1, \ldots R_n\}) \sim (C,\{R'_1, \ldots R'_k, R^\dag\}),$$
		where $R^\dag = C \setminus \bigcup_{i \leq k} R'_i$. We make use the fact for $i > k$, $\mu(W_{C,R_i}) = 0$ on account of the fact that $W_{C,R_i} \subseteq W_{A, P_i}$. Thus we have 
		\begin{align*}
			(A,\{P_1, \ldots P_n\}) \sim (C,\{R_1, \ldots R_n\}) \sim (C,\{R'_1, \ldots R'_k, R^\dag\})
		\end{align*}
		A symmetric argument ensures that also $(B,\{Q_1, \ldots Q_m\}) \sim C,\{R'_1, \ldots R'_k, R^\dag\}),$
		and so the two experiments are indifferent, as is required.
	\end{tproof}

	\begin{tproof}{thm:ent_rep}
		As before let $\vnm$ be the utility index that represents $\s$, that exists by \r{ax:eu}, normalized such that $\vnm(\{x\},\{\{x\}\}) = 0$; by \r{ax:sym}, the choice of $x$ is irrelevant, and by \r{ax:mon} the trivial experiment is the worst possible, so $\vnm: \E \to \R_{+}$ takes only weakly positive values.
		
		Let $\prob$ denote the set of finitely valued probability distributions, i.e., finite lists taking values in $[0,1]$ whose entries sum to 1. Let $\prob^* \subset \prob$ denote those whose entries are all strictly positive. Define $\zeta: \E \to \prob$ as $\zeta(A,\{P_1, \ldots P_n\}) = (\mu(W_{A,P_1}), \ldots, \mu(W_{A,P_n}))$.
		
		\begin{lemma}
			\label{lemma:constProb}
			For each $\{p_1, \ldots p_n\} \in \prob^*$, there exists some $(A,\P) \in \E$ such that $\zeta(A,\P) = \{p_1, \ldots p_n\}$. Moreover, if $\zeta(A,\P) = \zeta(B,\Q)$ then $(A,\P) \sim (B,\Q)$.
		\end{lemma}
		
		\begin{subproof}
			It is well know that since $\mu$ is non-atomic, there exists $\{W_1, \ldots W_n\} \in \part(\U)$, such that $\mu(W_i) = p_i$ for $i \leq n$ (for example, see \cite{billingsley1995probability} Problem 2.19(d)). By richness, there exists some $(A,\{P_1, \ldots, P_m\})$ such that $\{W_{A,P_i}\}_{i \in m}$ is $\mu$-equivalent to $\{W_1, \ldots W_n\}$. By Lemma \ref{lem:bijection}, it is without loss of generality to assume $\mu(W_{A,P_i}) = \mu(W_i)$ for $i \leq n$; it follows that $\mu(W_{A,P_j}) = 0$ for $j > n$. Then $(A, \{P_1 \cup \bigcup_{j > n} P_j, P_2, \ldots P_n\}$ is the desired experiment. 
			The later claim follows directly from \r{ax:sym}.
		\end{subproof}
		
		In light of Lemma \ref{lemma:constProb}, we can define the functional $\eta: \prob^* \to \R$ via 
		$$\eta(p_1, \ldots, p_n) = \vnm(A,\P),$$ 
		where $(A,\P) \in \zeta^{-1}(p_1, \ldots, p_n)$. Extend $\eta$ to all of $\prob$ by simply ignoring $0$s. That is, for each $(p_1, \ldots, p_n) \in \prob$, set $\eta(p_1, \ldots, p_n) = \eta(p_{k_1}, \ldots, p_{k_m})$, where $k_1, \ldots, k_m \subseteq 1, \ldots, n$ is the subsequence that selects strictly positive entries. Our normalization $\vnm(\{x\},\{\{x\}\}) = 0$ implies $\eta(1) = 0$.

		For $p = (p_1, \ldots p_n) \in \prob^*$ and $\{q^i\}_{i \leq n}$, where each $q^i = (q^i_1, \ldots, q^i_{m^i}) \in \prob^*$ is of (possibly distinct) length $m^i$, let 
		$$p\otimes \{q^i\}_{i \leq n} = (p_1q^1_1, \ldots, p_1q^1_{m^1}, \ldots, p_nq^n_1, \ldots, p_nq^n_{m^n}).$$
		We can intuitively view $p\otimes \{q^i\}_{i \leq n}$ as the reduction of a compound lottery over over $\sum_{i \leq n} m^i$ outcomes, thinking of $p$ as the marginal on $n$ first stage lotteries and $q^i$ the conditional lottery on $m^i$ outcomes given the realization of $p$.
		
		We will now show that $\eta$ satisfies the following three properties:
		
		\begin{enumerate}[label=(\texttt{K\arabic*})]
			\item \label{k:null}$\eta(p_1, \ldots, p_n) = \eta(p_1, \ldots, p_n, 0)$
			\item \label{k:sym}$\eta(p_1, \ldots, p_n) \leq \eta(\tfrac1n, \ldots, \tfrac1n)$
			\item \label{k:add}$\eta(p\otimes \{q^i\}_{i \leq n}) = \eta(p) + \sum_{i \leq n} p_i \eta(q^i)$
		\end{enumerate}
		
		Property \ref{k:null} follows immediately from the construction of $\eta$, in particular how it is extended from $\prob^*$ to $\prob$.  \ref{k:sym} follows immediately from \r{ax:sym}. We will show \ref{k:add}.
		
		Fix some $p \in \prob^*$ and $\{q^i\}_{i \leq n}$, with each $q^i \in \prob^*$. For each $0 \leq k \leq n$, let 
		\[
		q^{i,k} =
		\begin{dcases*}
			q^{i} &if  $i \leq k$ \\
			(1) &otherwise
		\end{dcases*}
		\]
		Notice that $q^{i,n} = q^i$ and $p\otimes \{q^{i,0}\} = p$. Thus, the result follows by showing that 
		\begin{equation}
			\label{eq:telescope}
			\eta(p\otimes \{q^{i,k}\})  = \eta(p\otimes \{q^{i,k-1}\}) + p_k \eta(q^k),
		\end{equation}
		for $0 < k \leq n$.
		
		For for $1 \leq i \leq n$ and $1 \leq j \leq m^i$, set $r^i_j = p_iq^i_j$.
		With this notation we can write the relevant distributions as 
		\begin{align*}
			p\otimes \{q^{i,k-1}\} =\, &(r^1_1, \ldots, r^1_{m^1}, \ldots, r^{k-1}_1, \ldots, r^{k-1}_{m^{k-1}}, p_{k}, p_{k+1}, \ldots p_n) \\
			p\otimes \{q^{i,k}\} =\, &(r^1_1, \ldots, r^1_{m^1}, \ldots, r^k_1, \ldots, r^k_{m^{k}}, p_{k+1}, \ldots p_n)
		\end{align*}
		
		From Lemma \ref{lemma:constProb}, we obtain some $e^k = (A,\{R^1_1, \ldots, R^k_{m^k}, P_{k+1}, \ldots, P_n\})$ in $\zeta^{-1}(p\otimes \{q^{i,k}\})$ and also some $e' = (B,\{Q_1, \ldots , Q_{m^k}\})$ in $\zeta^{-1}(q^k)$.
		Define $P_k = \bigcup_{j = 1}^{m^k} R^k_j$. By construction, $\mu(W_{A, P_k}) = \sum_{j = 1}^{m^k}\mu(W_{A, R^k_j}) = \sum_{j = 1}^{m^k} p_kq^k_j = p_k$. Thus, we have $e^{k-1} = (A,\{R^1_1, \ldots R^{k-1}_{m^{k-1}}, P_k,\ldots, P_n\})$ is in $\zeta^{-1}(p\otimes \{q^{i,k-1}\})$. Using the definition of $\eta$, we have
		\begin{align}
			\label{eq:entropfrompref}
			\begin{split}
				\eta(p\otimes \{q^{i,k-1}\}) =& \ \vnm(e^{k-1}) \\
				\eta(p\otimes \{q^{i,k}\}) =& \ \vnm(e^{k}) \\
				\eta(q^j) =& \ \vnm(e')
			\end{split}
		\end{align}
		
		At last, we have the requisite ingredients to appeal to \r{ax:ea}, and the representation via $\vnm$, obtaining
		$$
		\tfrac{1}{1 + \mu(W_{A, P_k})}\vnm(e^{k}) + \tfrac{\mu(W_{A, P_k})}{1 + \mu(W_{A, P_k})}0 = \tfrac{1}{1 + \mu(W_{A, P_k})}\vnm(e^{k-1}) + \tfrac{\\mu(W_{A, P_k})}{1 + \mu(W_{A, P_k}))}\vnm(e')
		$$
		Simplifying and plugging in the suitable replacements via \eqref{eq:entropfrompref} yields the desired relation.
		
		Theorem 1 of \cite{aleksandr1957mathematical} shows that if $\eta$ satisfies the properties
		\ref{k:null}--\ref{k:add}, it take the form 
		\begin{equation}
			\label{eq:log_rep}
			\eta(p_1, \ldots, p_n) = - \lambda \sum_{i=1}^n p_i \log(p_i),
		\end{equation}
		where $\lambda > 0$.\footnote{The property \ref{k:add} in \cite{aleksandr1957mathematical} is stated slightly differently: it requires all $q^i$ to be the same length, but allows for zero-probability entries. These formulations are clearly equivalent under \ref{k:null}, where 0s can be added to make each $q^i$ the same length.}
		Since we are free to rescale an expected utility representation by a positive constant (we only used a single degree of freedom in choosing the intercept $\vnm(\{x\},\{\{x\}\}) = 0$) we can set $\lambda = 1$.
		
		Finally, let $(A, \{P_1, \ldots, P_n\}) \in \E$. Without loss of generality, assume the first $k \leq n$ observations have positive probability (i.e., $\mu(W_{A,P_i}) > 0$ if and only if $i \leq k$). Set $P^\dag = P_1 \cup \bigcup_{i = k+1}^n P_i$
		We have
		\begin{align*}
			\vnm(A, \{P_1, \ldots, P_n\}) &= \vnm(A, \{P^\dag, P_2, \ldots, P_k\}) && \text{(from \r{ax:rcr})} \\
			&= \eta(\mu(W_{A,P_1}), \ldots, \mu(W_{A,P_k})) && \text{(definition of $\eta$)} \\
			&= \eta(\mu(W_{A,P_1}), \ldots, \mu(W_{A,P_n})) && \text{(from \ref{k:null})} \\
			&=  - \sum_{i \leq n} \log(\mu(W_{A,P_i}))\mu(W_{A,P_i}) , && \text{(from \eqref{eq:log_rep})} 
		\end{align*}
		as needed to complete the proof.
	\end{tproof}

	\newpage
	
	\bibliographystyle{aer}

	\singlespace
	\bibliography{EED.bib}
	
\end{document}